\documentclass[twocolumn,amsmath,amssymb,aps,pra,floatfix,nofootinbib,superscriptaddress,showpacs]{revtex4}

\usepackage{graphicx}
\graphicspath{{figures/pdf/},{figures/eps/}}
\usepackage{dcolumn}
\usepackage{bm}
\usepackage{amsthm}

\newtheorem{theorem}{Theorem}
\newtheorem{proposition}{Proposition}

\def\vec#1{{\bm #1}}

\def\ket#1{| #1 \rangle}
\def\bra#1{\langle #1 |}
\def\ip#1#2{\langle #1 | #2 \rangle}

\def\norm#1{\| #1 \|}

\def\RR{\mathbb{R}}

\def\dim{\operatorname{dim}}

\def\Span{\operatorname{span}}
\def\Reach{\operatorname{Reach}}

\def\Tr{\operatorname{Tr}}

\def\F{\mathcal{F}}

\def\H{\mathcal{H}}

\def\L{\mathfrak{L}}
\def\LL{\mathcal{L}}
\def\R{\mathcal{R}}

\def\RR{\mathbb{R}}

\def\uu{\mathbf{u}}
\def\su{\mathbf{su}}
\def\SU{\mathbf{SU}}
\def\UU{\mathbf{U}}

\def\sp{\mathbf{sp}}

\def\eps{\epsilon}

\begin{document}
\title{Global Control Methods for GHZ State Generation on 1-D Ising Chain}
\author{Xiaoting Wang}
\affiliation{Department of Applied Maths and Theoretical Physics, University of Cambridge,
Wilberforce Rd, Cambridge, CB3 0WA, United Kingdom}
\author{Abolfazl Bayat}
\affiliation{Department of Physics and Astronomy, University College
London, Gower St., London WC1E 6BT, United Kingdom}
\author{Sougato Bose}
\affiliation{Department of Physics and Astronomy, University College
London, Gower St., London WC1E 6BT, United Kingdom}
\author{Sophie G. Schirmer}
\affiliation{Department of Applied Maths and Theoretical Physics, University of Cambridge,
Wilberforce Rd, Cambridge, CB3 0WA, United Kingdom}

\date{\today}

\begin{abstract}
We discuss how to prepare an Ising chain in a GHZ state using a single
global control field only.  This model does not require the spins to be
individually addressable and is applicable to quantum systems such as
cold atoms in optical lattices, some liquid- or solid-state NMR
experiments, and many nano-scale quantum structures.  We show that GHZ
states can always be reached asymptotically from certain easy-to-prepare
initial states using adiabatic passage, and under certain conditions
finite-time reachability can be ensured.  To provide a reference useful
for future experimental implementations three different control
strategies to achieve the objective, adiabatic passage, Lyapunov control
and optimal control are compared, and their advantages and disadvantages
discussed, in particular in the presence of realistic imperfections such
as imperfect initial state preparation, system inhomogeneity and dephasing.
\end{abstract}


\maketitle

\section{Introduction}

Spin chains are an important theoretical model to understand properties
of many-body systems such as quantum phase
transitions~\cite{phase-transition}, and have become popular as a
possible model for quantum computation (QC) and
communication~\cite{Bose1}.  Various types of spin-spin interactions
including the Heisenberg, XY and Ising model have been discussed both
analytically and numerically.  Among these, the Ising model is one of
the most ubiquitous, arising in many different settings from atoms in
optical lattices~\cite{Jaksch,Cirac}, to NMR systems~\cite{Chuang} to
ion traps~\cite{Deng-Porras-Cirac} and polar
molecules~\cite{Micheli-Brennen-Zoller}.  For many of these systems
addressing individual spins selectively is extremely difficult, limiting
the type of control we can implement.  For instance, for cold atoms in
an optical lattice, addressing individual atoms with an external laser
is very difficult as the waist of the controlling laser is on the scale
of many lattice sites.  One way to circumvent this problem is by
eliminating the need for local addressing, i.e., by using only control
fields that act globally on all the spins at once.

One type of global control was first proposed in~\cite{Lloyd,Simon1},
where it was demonstrated that universal QC can be achieved by
controlling the qubits collectively, and since then extensive work has
been done on this approach~\cite{Simon-Bose,Lovett,Fitzsimons,
Fitzsimons07}. However, in all of these proposals, the spin-spin
interaction Hamiltonian and the globally controlled Hamiltonian are not
sufficient to realize universal QC, and additional resources are
necessary.  For example, in~\cite{Simon1}, the spin chain consists of
two types of qubits, $A$ and $B$, arranged in an alternating, repeating
pattern: $ABABAB\cdots$, and it is assumed that the two types of
spin-spin interactions $H_{AB}$ and $H_{BA}$ can be switched on and off
as needed, which is a demanding experimental requirement.  In an
improved scheme~\cite{Simon-Bose}, the coupling Hamiltonian can be kept
unchanged but the transition frequencies of each qubit must be tuned
individually, which is still too difficult to implement for cold atoms
in optical lattices at this time.  In another scheme~\cite{Lovett}, the
need for such ``individual tuning ability'' is avoided but at the
expense of requiring a chain with a repeating pattern of four types of
qubits, raising a high demand for the preparation process.  If the goal
is universal quantum computation then such extra requirements are
necessary as a controllability argument shows that the system
Hamiltonian together with only a global control Hamiltonian do not
generate the full Lie algebra $\su(2^N)$ but only a proper subalgebra,
rendering the system uncontrollable.  One interesting question therefore
is what further assumptions are necessary for universal QC, a question
that has been addressed in various recent
papers~\cite{Simon1,Simon-Bose,Lovett}. However, there are many tasks
that do not require controllability, and we can ask what interesting
tasks we can perform using global control only, without additional
resources.  This is the focus of this article.

In particular we demonstrate that global control alone is sufficient to
steer an Ising chain from a certain initial product state to a GHZ
state~\cite{GHZ}. Such GHZ states (or multi-qubit Cat states or NOON
states, as they are variously called) are of utmost importance for
improving frequency standards beyond the classical realm
\cite{Wineland}, and could result in highly sensitive magnetometers
\cite{Taylor}. This is why recently there has been extensive interest in
their experimental realization, for example, in ion traps
\cite{Leibfried} and with multiple nuclear spins in a molecule
\cite{Jones}.  However, a fast/ non-adiabatic (so as to be robust to
decoherence) method of generating these systems with minimal control,
such as global fields on an Ising chain, is still an open problem, and
will be a very important milestone for realizing quantum enhanced
sensing and standards.  We again emphasize that all control schemes
relying purely on global control are useful for experiments on the
systems where individual addressability is not available.  We show that
with global control, certain GHZ states can be reached in finite time
from a given, easy-to-prepare product state, i.e., that there always
exists a control that achieves the task, and consider and compare three
different methods to design suitable controls: adiabatic passage,
Lyapunov control design and optimal control. One of the most interesting
issue of our proposal is that in all three methods the initial and final
states are the eigenstates of the Hamiltonian and hence the output state
does not evolve anymore.  This makes it convenient since no fine-tuning
of the control time is necessary and the output state can be saved for
further tasks.

The article is organized as follows: In Sec.~II the model and control
problems are defined.  In Sec.~III, controllability of the system, or
rather lack of it, and reachability are discussed.  In Sec.~IV three
methods for GHZ state generation are considered in detail, namely,
adiabatic passage, Lyapunov control and optimal control.  Finally, the
effect of various procedural imperfections such as imperfect
initialization, inhomogeneity and decoherence are studied in Sec. V,
followed by a brief summary and discussion of the results in Sec VI.

\section{Ising model and Control Problem}

In the following we consider a 1-D spin chain of length $N$ with the
Hamiltonian
\begin{align}
 \label{eqn:Hf}
 H_f(t)= J[H_0+f(t)H_1], 
\end{align}
where, $H_0$ models the fixed interaction between neighboring spins with
a fixed strength $J$, and $H_1$ the control interaction, corresponding
to an applied external field $f(t)$.  We choose $H_0$ and $H_1$ such
that the total Hamiltonian~(\ref{eqn:Hf}) is a uniform nearest-neighbor 
Ising interaction in a transverse time-dependent magnetic field
\begin{equation}
\label{eqn:H}
  H_0=\sum_{n=1}^{N-1} Z_n Z_{n+1}, \quad
  H_1=\sum_{n=1}^{N} X_n,
\end{equation}
where $X,Y,Z$ are the Pauli matrices.  Practically, one can think of
$f(t)$ as a time-dependent global magnetic field in the $x$-direction
that causes all spins to rotate simultaneously, while all spins are
constantly coupled via Ising interaction.  The field $f(t)$ is varied
with respect to time $t$ and at time $t=0$ takes $f_0=f(0)$.  The
associated controlled dynamical evolution is given by the Schrodinger
equation
\begin{align}
 \label{control-dynamics} 
 \frac{d}{dt}\ket{\psi(t)} =-i H_f(t)\ket{\psi(t)}
\end{align}
where we have assumed units such that $\hbar=1$.

Our main objective is to prepare an Ising chain of length $N$ with
Hamiltonian~(\ref{eqn:Hf}) in one of the following GHZ states
\begin{subequations}
\label{eqn:ghz-state}
\begin{align}
 \ket{\psi_d^{(1)}} 
  &= \textstyle
     \frac{1}{\sqrt{2}}(\ket{0\cdots0}+\ket{1\cdots1})\\
 \ket{\psi_d^{(2)}} 
  &= \textstyle
     \frac{1}{\sqrt{2}}(\ket{0\cdots0}-\ket{1\cdots1})\\
 \ket{\psi_d^{(3)}} 
  &= \textstyle
      \frac{1}{\sqrt{2}}(\ket{0101\cdots}+\ket{1010\cdots})\\
 \ket{\psi_d^{(4)}} 
  &= \textstyle
      \frac{1}{\sqrt{2}}(\ket{0101\cdots}-\ket{1010\cdots}),
\end{align}
\end{subequations}
starting in the ground state $\ket{\psi_0^{(f_0)}}$ of the Hamiltonian
$H_{f}(0)$, i.e., our aim is to find a magnetic field $f(t)$ such that
the system states evolves into one of the entangled GHZ states given in 
Eq. (\ref{eqn:ghz-state}) under the action of the resulting Hamiltonian, 
starting from the initial state  $\ket{\psi_0^{(f_0)}}$. 

In the absence of the magnetic field, $f_0=0$, the ground state of the
$H_f(0)$ is two-fold degenerate, spanned by
$\{\ket{\psi_d^{(1)}},\ket{\psi_d^{(2)}}\}$ for $J<0$, and
$\{\ket{\psi_d^{(3)}},\ket{\psi_d^{(4)}}\}$ for $J>0$.  Although the
target states are ground states of the Hamiltonian $H_0$, due to the
degeneracy cooling alone does not suffice to prepare the system in any
of the states~(\ref{eqn:ghz-state}).  Rather, simple cooling will result
in the system being left in a mixture of different ground states, which
is not useful.  On the other hand, in the presence of the magnetic field
the ground state of the total Hamiltonian $H_f(0)$ with $f_0\ne 0$ is
non-degenerate, and if the system is cooled in the presence of a global
field $f_0$ along the $x$-axis, it will be initialized in the ground
state $\ket{\psi_0^{(f_0)}}$ of $H_f(0)$.  As $|f_0|\to\infty$ we have
\begin{subequations}
\begin{align}
  \lim_{Jf_0\to-\infty} \ket{\psi_0^{(f_0)}} = \ket{+\ldots+},\\
  \lim_{Jf_0\to+\infty} \ket{\psi_0^{(f_0)}} = \ket{-\ldots-},
\end{align}
\end{subequations}
where $\ket{\pm}=\frac{1}{\sqrt{2}}(\ket{0}\pm\ket{1})$, and $\ket{0}$
and $\ket{1}$ are the eigenstates of $Z$, i.e., for sufficiently large
$|f_0|$ we can assume $\ket{\psi_0^{(f_0)}}$ to be approximately equal
to $\ket{\psi_0^{+}}=\ket{+\ldots+}$ or
$\ket{\psi_0^{-}}=\ket{-\ldots-}$.  Choosing $|f_0|$ large also ensures
a sufficiently large energy gap for efficient cooling.

\section{Symmetries, Non-Controllability, Reachability}

The dynamics~(\ref{control-dynamics}) can be expressed in terms of the
unitary process $U(t)$ satisfying
\begin{equation}
 \label{eqn:SEU}
  \dot U(t)  = -i H_f(t)U(t)
\end{equation}
with $U(0)=I$.  Denoting the solution of (\ref{eqn:SEU}) for a given
control $f(t)$ by $U(t,f(t))$, the reachable set of unitary operators
$\R$ is defined as the set of unitary matrices $U(t,f)$ that can be
generated by the dynamics~(\ref{eqn:SEU}) in a finite time $t_f$ for
some admissible control $f\in \F$, i.e., $\bar{U}\in \R$ if and only if
there exists an admissible control $f(t)$ such that
$U(t_f,f(t))=\bar{U}$.  From control theory, we have~\cite{Jurdjevic}:

\begin{theorem}
\label{reachable_set} Let $\L$ be the Lie algebra generated by
$\Span_{f\in\F}\{J[H_0+f(t)H_1]\}=\Span\{H_0,H_1\}$, called the dynamical
Lie algebra.  Then $\R=e^{\L}$.
\end{theorem}

Hence, if $\L=\su(N)$ or $\L=\uu(N)$ then we have $\R=\SU(N)$ or
$\R=\UU(N)$ , the system is controllable in that we can implement any
unitary operator $U$ up to at most a global phase $e^{i\phi}$.  In this
case any pure state is reachable from any other pure state, and more
generally, any two density operators with the same spectrum can be
interconverted.  Since the spectrum of any operator is preserved under
unitary evolution, this is the most we can hope for.  In this sense,
controllability is a sufficient condition for reachability.  On the
other hand, any system that possess symmetries will not be fully
controllable.  For the Hamitonian $H_f(t)$, if there exists a Hermitian
operator $M$ such $[M,H_m]=0$ for $m=0,1$, then the Hamiltonians $H_0$
and $H_1$ are simultaneously block-diagonalizable. In this case the
Hilbert space can be decomposed into orthogonal invariant subspaces
$\H_k$ such that any initial state $\ket{\psi(0)} \in \H_k$ remains in
$\H_k$ under the evolution, regardless of what control $f(t)$ we apply.
Since no states outside $\H_1$ can be reached from an initial state in
$\H_1$, decomposability immediately implies non-controllability.
However, a target state $\ket{\psi_d}$ may still be reachable from an
initial state $\ket{\psi(0)}$ if both belong to the same subspace.  In
particular, this is the case if the system is controllable on the
relevant invariant subspace.

Applying these results to our Ising chain subject to global control we
see immediately that the system~(\ref{eqn:H}) possesses symmetries as
both Hamiltonians $H_0$ and $H_1$ commute with the ``$X$-parity''
operator $M=\prod_{n=1}^N X_n$.  Hence, $H_0$ and $H_1$ are
simultaneously block-diagonalizable.  In our case $M$ has two
eigenspaces with eigenvalues $\pm 1$, spanned by
\begin{equation}
  \H_\pm = \Span
  \left\{\frac{1}{\sqrt{2}}(\vec{e}_{_k}\pm \vec{e}_{2^{N-k+1}})\right\},
 \quad \forall k=1,\ldots,2^{N},
\end{equation}
where $\vec{e}_k$ is the basis vector with $1$ in the $k$th position and
all other entries $0$; changing from the basis $\{\vec{e}_k\}$ to an
eigenbasis of $M$ simultaneously block-diagonalizes $H_0$ and
$H_1$. This shows that the Ising chain with global control is not
controllable and explains why additional resources are required to
obtain universal
QC~\cite{Simon1,Simon-Bose,Lovett,Fitzsimons,Fitzsimons07}.  However,
there are many useful tasks that can be performed under the
evolution~(\ref{control-dynamics}).

For $N>2$ there are further symmetries and both blocks are further
decomposable.  There are various approaches to decompose the Hilbert
space into invariant subspaces that are not further decomposable.  One
approach is to proceed as before and find the symmetry operators $M$ on
each subspace.  If $V$ is an eigenbasis of $M$ then $\tilde{H}_m=V^\dag
H_m V$ will be block-diagonal.  This can be done recursively until no
further symmetries are found for any of the blocks, leaving us with
indecomposable blocks.  This approach becomes tedious, however, when
there are many symmetries.  Alternatively, we can calculate the
eigenvectors of a linear combination of the Hamiltonians, $H=\alpha
H_0+\beta H_1$.  Letting $V$ be a unitary matrix whose columns are the
normalized eigenvectors of $H$, we define an adjacency matrix
$A=(a_{ij})$ with $a_{ij}=1$ if the absolute value of the $(i,j)$th
element of the matrix $V^\dag H_0 V$ is greater than some threshold
value $\delta$, and $0$ otherwise, and find the connected components of
$A$, which define the respective invariant subspaces.  The accuracy of
this approach depends on suitable choice of $\alpha$, $\beta$ and
$\delta$.  Choosing $\delta=10^{-8}$ and $\alpha=2$, $\beta=3$, we
calculated the subspace decomposition for Ising chains up to $N=14$.

Having found a decomposition of the system into indecomposable subspaces
the next step is to verify if the initial and target states both belong
to the same invariant subspace.  For $J<0$ we verified numerically that
both $\ket{\psi_0^+}$ and the GHZ state $\ket{\psi_d^{(1)}}$ belong to
the same invariant subspace for $N=2,\ldots,14$.  To establish
reachability the next step is usually to try to show that the system is
controllable on this invariant subspace.  For $N=2$, this is easy.
Changing the basis from $\{\ket{0},\ket{1}\}$ to $\{\ket{+},\ket{-}\}$
the Hamiltonians become
\begin{equation*}
 H_0 
     =\begin{bmatrix}
       0 & 0 & 0 & 1 \\
       0 & 0 & 1 & 0 \\
       0 & 1 & 0 & 0 \\
       1 & 0 & 0 & 0
      \end{bmatrix}, \quad
 H_1 
     =\begin{bmatrix}
       -2 & 0 & 0 & 0\\
       0 & 0 & 0 & 0 \\
       0 & 0 & 0 & 0 \\
       0 & 0 & 0 & 2
      \end{bmatrix}
\end{equation*}
and similarly
\begin{align*}
 \ket{\psi_0^+} &=[1,0,0,0]^T, \quad
 \ket{\psi_d^{(1)}} &=[1,0,0,1]^T/\sqrt{2}.
\end{align*}
This clearly shows that $\ket{\psi_0^{+}}$ and $\ket{\psi_d^{(1)}}$
belong to a two-dimensional subspace $\H_s$ spanned by the basis vectors
$\vec{e}_1$ and $\vec{e}_4$, on which we have $H_0^{(s)}=X$ and
$H_1^{(s)}=-2Z$, showing that the dynamical Lie algebra generated is
$\SU(2)$.  Hence, the system is controllable on this subspace and
$\ket{\psi_d^{(1)}}$ is reachable from $\ket{\psi_0^{+}}$.

A similar approach allows us to establish reachability of the target
state from the initial state for $N=3$.  One might therefore hope that
the the system is controllable on the relevant invariant subspace $\H_s$
for all $N$.  Unfortunately, this is not true for higher dimensions.
For $N=4$, for instance, the smallest subspace $\H_s$ that contains both
$\ket{\psi_0^{+}}$ and $\ket{\psi_d^{(1)}}$ and is invariant under the
Hamiltonian has dimension $6$, while the dynamical Lie algebra $\L$
generated by $H_0$ and $H_1$ on the entire space has only dimension
$\dim{\L}=16$.  This is strictly smaller than the dimension needed for
controllability on the subspace, which is $35=\dim\su(6)$ for full
controllability and $21=\dim\sp(3)$ for pure-state controllability for a
six-dimensional subspace\footnote{Density operator controllability on a
subspace of dimension $N$ requires that we can generate the full Lie
algebra of (trace-zero) Hermitian matrices $\uu(N)$ ($\su(N)$); if we
are only interested in pure-state controllability it suffices if we can
generate symplectic Lie algebra $\sp(N/2)$. See \cite{controllability}}.
Hence, the system cannot be controllable on $\H_s$.  Explicit
calculations for various $N$ suggest that dynamical Lie algebra on the
entire Hilbert space is a $2^N$-dimensional reducible representation of
$\uu(N)$, and that the dynamical Lie algebra on the smallest invariant
subspace that contains both $\ket{\psi_0^{+}}$ and $\ket{\psi_d^{(1)}}$
is an irreducible representation of $\su(N)$ or $\uu(N)$.  As for $N>3$
the dimension of the subspace $\H_s$ is greater than $N$, this implies
non-controllability on $\H_s$ for $N>3$.  Therefore, in general a
different approach is needed to show that $\ket{\psi_d^{(1)}}$ is
reachable from $\ket{\psi_0^{+}}$.  Similar problems arise when trying
to assess the reachability or non-reachability of the other GHZ states
$\ket{\psi_d^{(k)}}$ from $\ket{\psi_0^{+}}$ or $\ket{\psi_0^{-}}$.

\section{Control Methods for GHZ Generation}

We now discuss three different methods for generating GHZ states
(\ref{eqn:ghz-state}) from the separable states $\ket{\psi_0^{+}}$ or
$\ket{\psi_0^{-}}$: (i) adiabatic passage; (ii) Lyapunov control and
(iii) the optimal control theory.  We shall see that adiabatic passage
demonstrates asymptotic reachability of certain GHZ states from certain
product states.  It is also has many benefits in that a simple field can
achieve high fidelity and rather robust population transfer for spin
chains of varying length.  A general drawback of adiabatic schemes,
however, is that the target state is exactly reachable only in the limit
$t\to +\infty$, and the time required to prepare the target state with
sufficiently high fidelity can be long.  This prompts the question
whether we could do better using some form of optimal control design
either Lyapunov control or global optimal control.  

\subsection{Adiabatic Passage}

Assume we initialize the system in the ground state
$\ket{\psi_0^{(f_0)}}$ of $H_{f}(0)=J[H_0+f_0H_1]$ for some $f(0)=f_0$.
If we can show that a particular GHZ state in (\ref{eqn:ghz-state}) is
the limit of the ground state of $H_f(t_f)$ as $f(t_f)\to 0$, then it is
possible to adiabatically transfer the system to one of the GHZ states
(\ref{eqn:ghz-state}).  Indeed, such a result has been implicitly shown
for an Ising ``chain'' with periodic boundary
conditions~\cite{Stelmachovic} using the Jordan-Wigner transformation,
and we can easily show directly that this result is true for proper
chains using the fact that $H_0$ and $H_1$ simultaneously commute with
the $X$-parity operator $M$, and thus that $[H_f(t),M]=0$ regardless of
the choice of $J$ and $f(t)$.  We clearly have
\begin{subequations}
 \begin{align}
 M\ket{\psi_0^{+}} = (+1)^N\ket{\psi_0^{+}}, \\
 M\ket{\psi_0^{-}} = (-1)^N\ket{\psi_0^{-}}.
\end{align}
\end{subequations}
This shows that $\ket{\psi_0^{+}}$ always has positive parity, while
$\ket{\psi_0^{-}}$ has positive parity for $N$ even, and negative parity
for $N$ odd. The same must hold for the finite values of $f_0$, i.e.,
the eigenstate $\ket{\psi_{0}^{(f_0)}}$ has positive parity if $f_0<0$,
or $f_0>0$ and $N$ even, and negative parity if $f_0>0$ and $N$ odd.  As
the parity is a conserved quantity, the adiabatic limit state
$\ket{\psi(t_f)}$ must have the same parity as $\ket{\psi_0^{(f_0)}}$.

It is easy to see that the intersection of the two-fold degenerate
ground state manifold of $JH_0$ with the $+1$ ($-1$) eigenspace of $M$
is unique.  For $J<0$ the ground state manifold of $JH_0$ is spanned by
the GHZ states $\{\ket{\psi_d^{(k)}}:k=1,2\}$, and it is easy to see
that $\ket{\psi_d^{(1)}}$ has positive, and $\ket{\psi_d^{(2)}}$
negative parity.  For $J>0$ the ground state manifold of $JH_0$ is
spanned by the GHZ states $\{\ket{\psi_d^{(k)}}:k=3,4\}$, and it is easy
to see that $\ket{\psi_d^{(3)}}$ has positive, and $\ket{\psi_d^{(4)}}$
negative parity.  Thus we have
\begin{subequations}
\begin{align}
  \ket{\psi_d^{(1)}} &\in \Reach\{\ket{\psi_0^{(f_0)}}: J<0, Jf_0<0, \mbox{
or}\nonumber\\
                     & \hspace{1.05in} J<0, Jf_0>0, N \mbox{ even}\},\\
  \ket{\psi_d^{(2)}} &\in \Reach\{\ket{\psi_0^{(f_0)}}: J<0, Jf_0>0, N
 \mbox{ odd}\},\\
  \ket{\psi_d^{(3)}} &\in \Reach\{\ket{\psi_0^{(f_0)}}: J>0, Jf_0<0,
 \mbox{ or}\nonumber \\
                     & \hspace{1.05in} J>0, Jf_0>0, N \mbox{ even}\},\\
  \ket{\psi_d^{(4)}} &\in \Reach\{\ket{\psi_0^{(f_0)}}: J>0,
  Jf_0>0, N \mbox{ odd}\}.
\end{align}
\end{subequations}

Although adiabatic passage provides a way to drive the system to a
certain GHZ state $\ket{\psi_d}$ in (\ref{eqn:ghz-state}), strictly
speaking it only implies asymptotic reachability of $\ket{\psi_d}$ from
the given initial state $\ket{\psi_0^{(f_0)}}$ as the Adiabatic Theorem
provides that the error between the exact final state $\ket{\psi(t_f)}$
and $\ket{\psi_d}$ will go to zero only as $t_f\to+\infty$, i.e., that
the fidelity
\begin{equation}
   F(t) = |\ip{\psi_d}{\psi(t)}|^2 \to 1 \mbox{ as } t\to +\infty,
\end{equation}
if the field $f(t)$ changes sufficiently slowly so that the rate of the
change of the ground state energy $\eps_1$, is small compared to the
energy gap between the ground state and the first excited state, i.e.,
$\Delta\eps=\eps_2-\eps_1$.  This is not really a problem in practice as
we are likely to be satisfied if we can get sufficiently close to the
target state in a finite time, and similar adiabatic schemes have indeed
been proposed, e.g., for strings of neutral atoms in~\cite{atom-Zoller}.
Moreover, under certain conditions asymptotic reachability implies
finite-time reachability.  

\begin{proposition}
\label{appendix} If $e^\L$, where $\L$ is the dynamical Lie algebra,
is compact then asymptotic reachability of $\ket{\psi_d}$ from an
initial state $\ket{\psi(0)}$ implies that $\ket{\psi_d}$ is
reachable from $\ket{\psi(0)}$ in finite time.
\end{proposition}

\begin{proof}
If $\ket{\psi_d}$ is asymptotically reachable from $\ket{\psi(0)}$
then there exists a path $\ket{\psi(t)}$ such that
$\lim_{t\to+\infty}\ket{\psi(t)}=\ket{\psi_d}$.  We can choose a
time sequence $\{t_n\}$ with $t_n\to +\infty$ such that
$\ket{\psi(t_n)}=U_n\ket{\psi(0)}\to \ket{\psi_d}$, where $U_n\in
e^\L$ is a unitary process.  Since $\{U_n\}$ is a sequence in the
compact group $e^\L$, there exists a converging subsequence
$\{U_{n_k}\}$ such that $U_{n_k}\to\bar U\in e^\L$, satisfying
$U\ket{\psi(0)}=\ket{\psi_d}$.  From Theorem~\ref{reachable_set},
$\bar U\in e^\L$, i.e., there exists a dynamical trajectory during
$[0,T]$ for some finite $T$ such that $\ket{\psi(T)}=\ket{\psi_d}$.
\end{proof}

According to~\cite{dalessandro} the dynamical Lie group is the direct
product of an Abelian Lie group and a semi-simple compact Lie group.
Numerical computations of the generators of the Abelian Lie group for
various $N$ suggest that in our case, the Abelian group is compact, and
thus that $e^\L$ generated by $H_f(t)$ is compact, and hence we have
finite-time reachability.

\begin{figure*}
\includegraphics[width=\columnwidth]{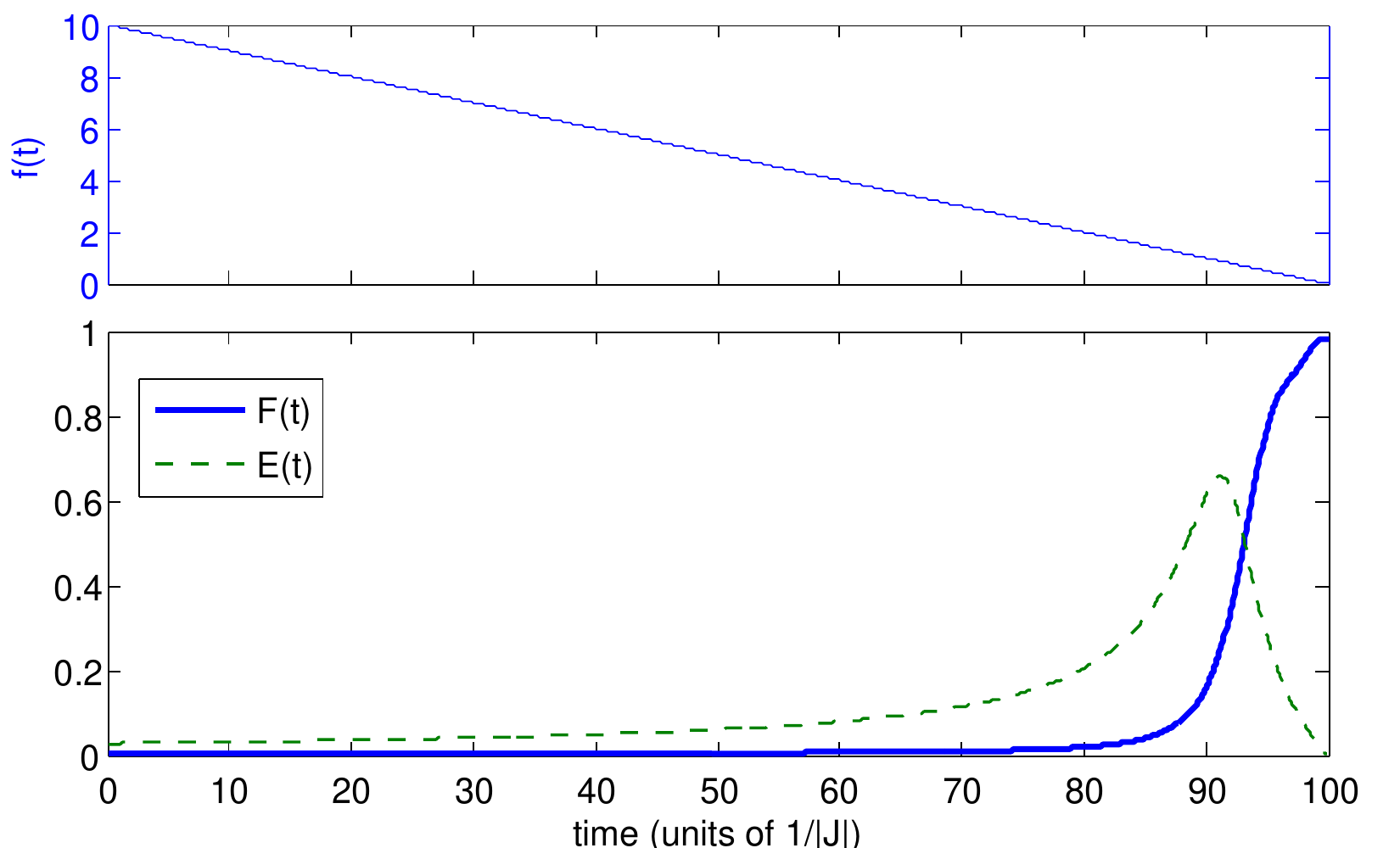}
\includegraphics[width=\columnwidth]{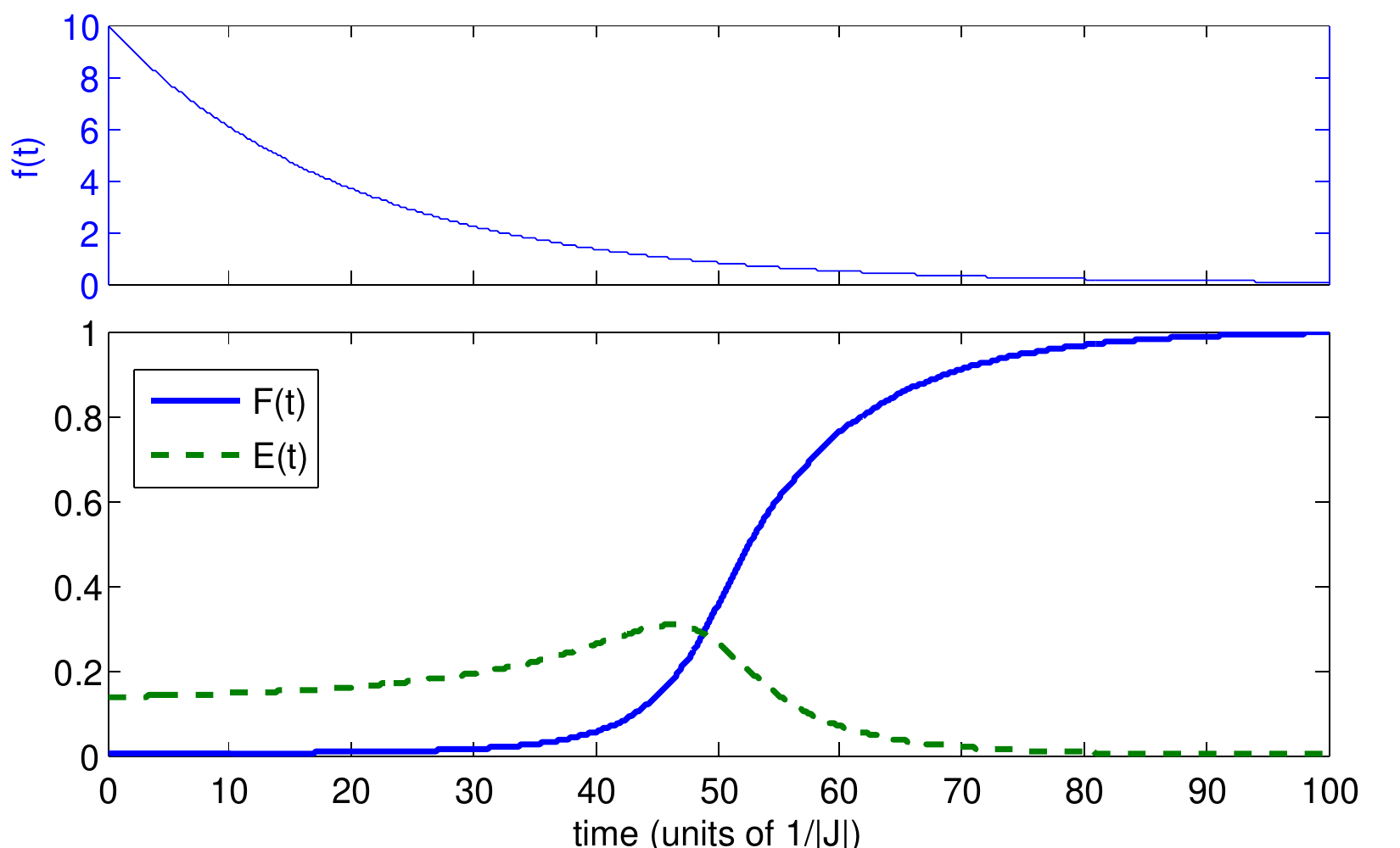}
\caption{(Color online) Adiabatic control (top) and corresponding
relative energy gap $E(t)=\dot{\eps}_1(t)/\Delta\eps(t)$ and fidelity
$F(t)$, for a linearly varying field, $f(t)=10(1-t/100)$, left, and and
exponentially varying control, $f(t)=10e^{-0.05t}$, (right).  The
fidelity $F(t)$ corresponding to the projection onto the GHZ state
$\ket{\psi_d^{(1)}}$ in both cases asymptotically approaches a limiting
value of $\approx 1$, but for the linear control it increases sharply
near the end of the pulse.  For the exponentially decaying control field
the increase begins much sooner, is more gradual, and higher fidelities 
can be achieved in shorter time.}
\label{fig:pulse_F}
\end{figure*}

\begin{figure}
\includegraphics[width=\columnwidth]{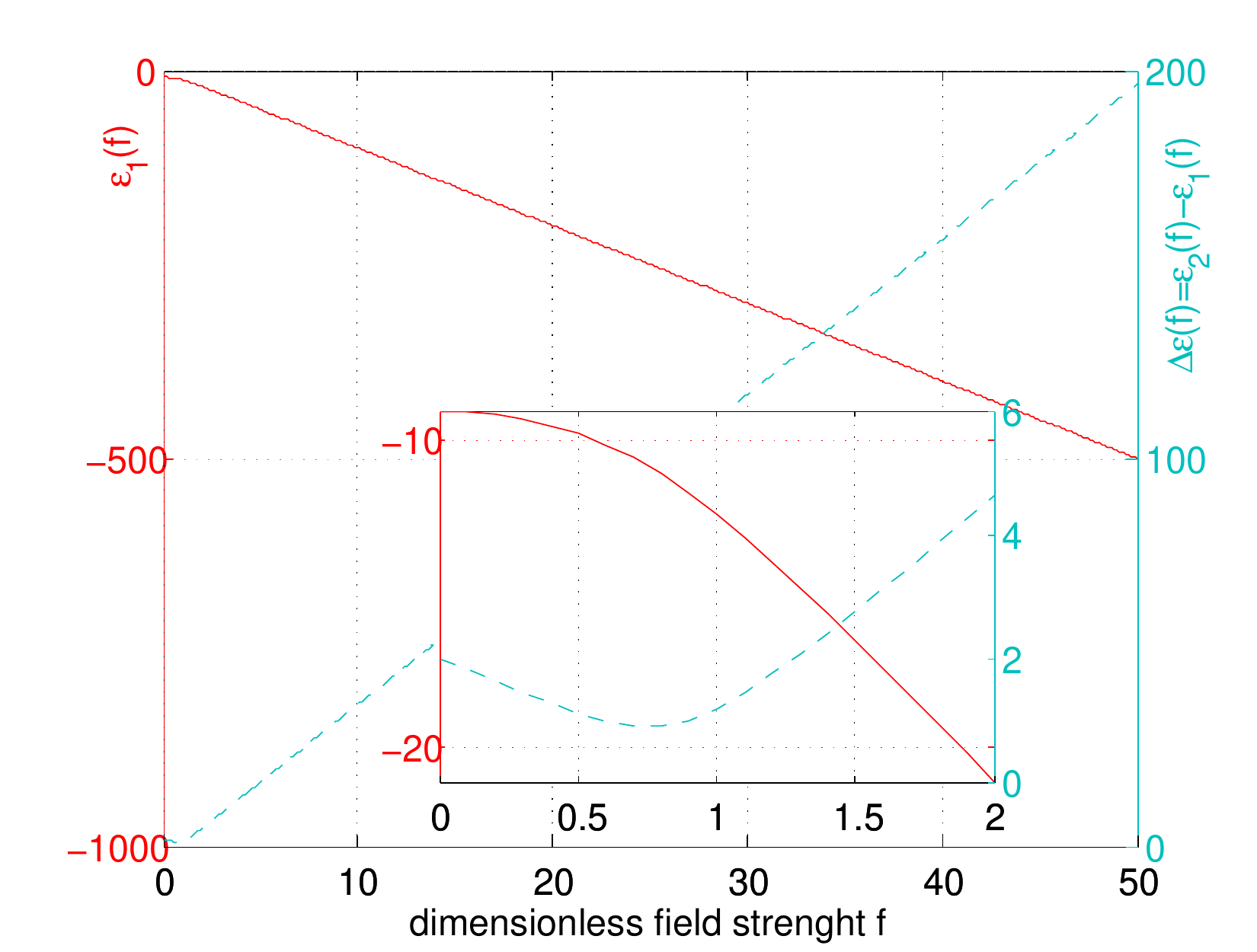}
\caption{(Color online) Ground state energy $\eps_1$ (solid line) and
energy gap (dashed) $\Delta\eps=\eps_2-\eps_1$ for a chain of length
$10$ with Hamiltonian $H_f$, given in Eq. (\ref{eqn:Hf}), and $J<0$ as a
function of the applied field $f$, for the dynamics restricted to the
invariant subspace $\H_s$ containing the GHZ state $\ket{\psi_d^{(1)}}$.
For large fields the $H_1$ term dominates and the ground state energy
and energy gap depend approximately linearly on the field: $\eps_1(f)
\approx -10f$, $\Delta\eps(f)\approx 4 f$.  For $0<f<2$ the inset
suggests an approximately quadratic dependence of both the ground state
energy and energy gap on the applied field.
} \label{fig:adia_gap}
\end{figure}

To assess the performance of adiabatic passage we turn to simulations.
Although the choice of $f(t)$ does not matter in theory, provided it
varies sufficiently slowly and vanishes as $t\to+\infty$, in practice we
are usually interested in preparing a sufficiently close approximation
to the target state in as little time as possible, and in this case the
choice of $f(t)$ does matter as a comparison of two simple controls, a
linearly decreasing field $f(t)=f_0(1-t/t_f)$ for $t\in [0,t_f]$, and a
decaying exponential, $f(t)=f_0 e^{-\mu t}$, in Fig.~\ref{fig:pulse_F}
shows.  The results of these simulations suggest that the latter choice
is preferable in terms of speed and robustness.  This can partly be
explained by comparing $E(t)=\dot{\eps}_1/\Delta\eps$.  For the linear
field $E(t)$ is negligible for most of the pulse duration and spikes
towards the end of the pulse, mirroring the sharp increase in the
population of the target state near the target time
(Fig.~\ref{fig:pulse_F}, left).  For the exponentially decaying field
$f(t)$ drops to $f\approx 1$ much faster and spends more time in the
region $0<f<1$, where most of the interesting evolution takes place.
Also, for the exponential field we have $\dot{f}(t)/f(t)=\mu$, i.e.,
$E(t)\approx -2.5\mu$ is approximately constant until the field has
dropped to $f\approx 1$, while for the linear field
$\dot{f}(t)/f(t)=\mu/f(t)$, where $\mu$ is the slope, and thus
$E(t)\approx -2.5\mu/f(t)$ will be negligibly small for $f(t)$ large.

In both cases the choice of $f_0$ is dictated by practical concerns.
Fig.~\ref{fig:adia_gap} shows the ground state energy $\eps_1$ and the
energy gap $\Delta\eps$ for $H_f(t)$ with $J<0$, as a function of field
strength $f$. It shows that for $f>1$ the energy gap increases with $f$.
Assuming the rate of cooling to be proportional to the energy gap, this
compels us to choose $f_0$ as large as possible to ensure efficient
cooling and initial state preparation.  It also ensures that the initial
state is close to the desired initial state $\ket{\psi_0^+}$ as the
error $1-|\ip{\psi_0^+}{\psi_0^{(f_0)}}|^2$ decreases quadratically in
$f_0$.  At the same time, to maintain adiabaticity, the field must
decrease slowly, and thus the time required for the field to decay to a
certain value close to zero increases proportionally.  For fixed $f_0$
the asymptotic value of the fidelity depends on the decay rate $\mu$.
In general it decreases as $\mu$ increases.  At the same time, the time
required to reach a certain target fidelity (below the asymptotic value)
decreases with increasing $\mu$.  Thus, if we are only interested in
achieving a certain target fidelity of say 99\%, there will be an
optimal value of the decay rate $\mu$ that achieves 99\% transfer the
shortest amount of time.  Fig. \ref{fig:adia_error} shows the log-error
$\log_{10}(1-F)$ for chains for different length $N$ for an
exponentially decaying field with $\mu=0.1$.  We note that the time when
the fidelity crosses the 99\% threshold remains in a narrow range of
$[44,49]$ for a significant range of $N$.  The onset of oscillations in
the evolution of the fidelity (population of the target state) for
larger $N$ suggests that non-adiabatic effects arise for this $\mu$,
resulting in population transfer from the (instantaneous) ground state
$\ket{\psi_1(t)}$ of $H(t)$ to excited states.  This is confirmed by the
population plot $1-p_1(t)$ for $N=11$ (Fig.~\ref{fig:adia_error},
inset), and shows that we must reduce the decay rate $\mu$ for longer
chains to maintain adiabatic evolution.

\begin{figure}
\includegraphics[width=\columnwidth]{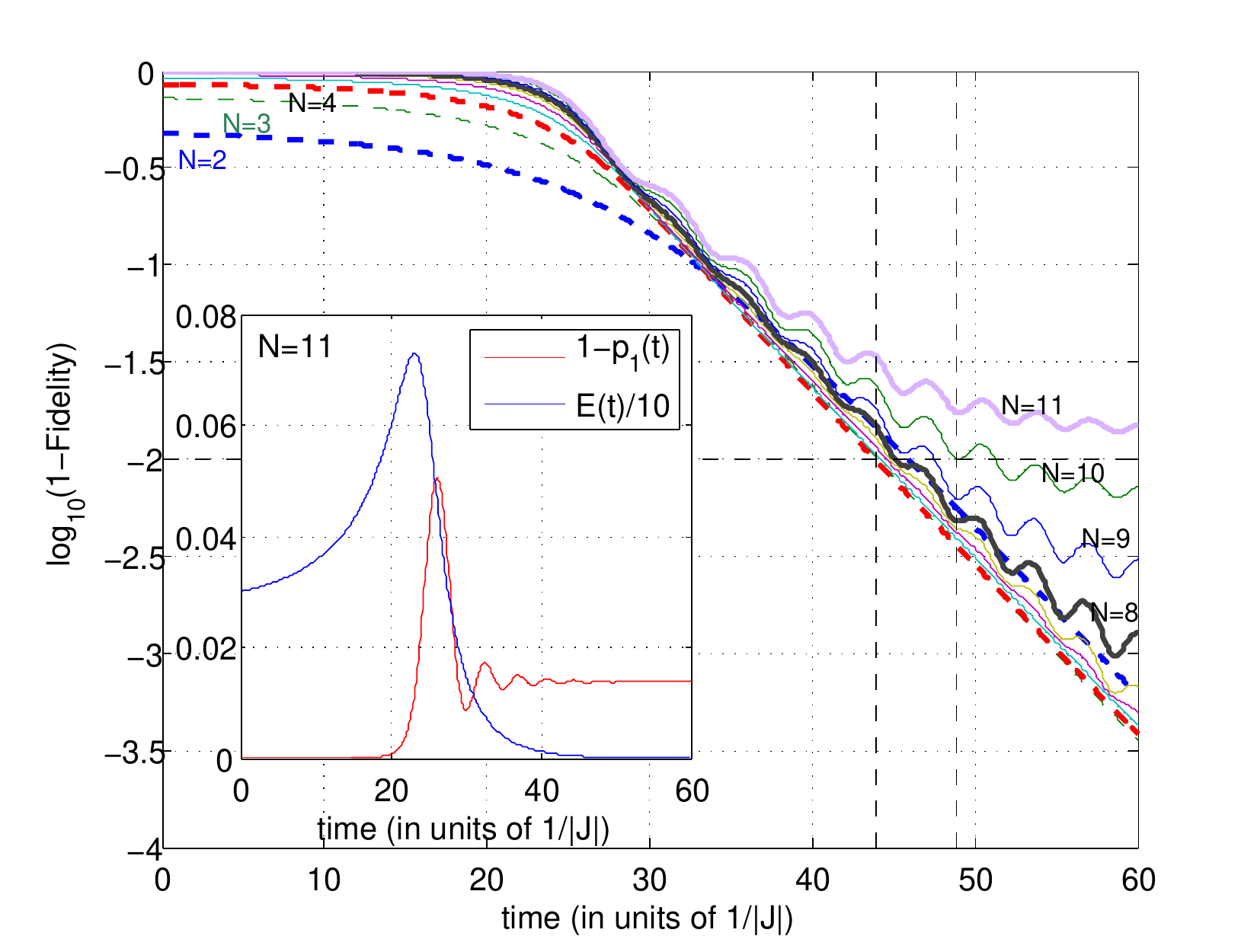} \caption{(Color
online) The log-error, i.e., log of population not transferred to the
target state, as function of time for Ising chains (\ref{eqn:Hf}) with
$J<0$ of length $N=2,\ldots 11$ starting in $\ket{\psi_0^{(f_0)}}$,
shows that the field $f(t)=10e^{-0.1t}$ achieves 99\% population
transfer to the target state for chains of length $N=2$ to $10$ in a
narrow time window $t_{0.99}\in[44,49]$.  For larger $N$ oscillations
indicate non-adiabatic effects, which prevent reaching 99\% fidelity for
$N=11$ with this field.  The population of the upper levels $1-p_1(t)$
(inset) clearly shows that the pulse induces some population transfer to
these states, occuring, as expected, following the peak in $E(t)$,
indicating that the peak value is too large to fully maintain
adiabaticity.}  \label{fig:adia_error}
\end{figure}

\subsection{Lyapunov control}

A simple way to solve certain optimal control problems is construct
a Lyapunov function for the dynamics (\ref{control-dynamics}).  A
natural candidate for a Lyapunov function is a monotonic function of
the Hilbert-Schmidt distance for density operators such as
\begin{equation} \textstyle
 V(\rho,\rho_d) = \frac{1}{2}\norm{\rho-\rho_d}^2
                = \frac{1}{2}\Tr[(\rho-\rho_d)^2].
\end{equation}
satisfying $V\ge 0$ which equality holds if and only if
$\rho_d=\rho$. The essential idea of Lyapunov control is to design
appropriate control dynamics such that $V$ becomes a Lyapunov function,
i.e., $V$ keeps decreasing along every trajectory of $\rho(t)$.  This
can be realized by choosing
\begin{equation}
\label{eqn:Lya-control}
  f(t) = f(\rho(t),\rho_d) = \kappa \Tr([iH_1,\rho_d]\rho(t)).
\end{equation}
Then for $V(t)=V(\rho(t),\rho_d)$, we have
\begin{equation}
  \dot V(t) = -f(t)\Tr([iH_1,\rho_d]\rho(t)) = -\kappa f(t)^2\le 0,
\end{equation}
i.e., $\rho(t)$ evolves towards $\rho_d$.  Ideally, if $V(t)\to 0$ as
$t\to+\infty$, we have $\rho(t)\to\rho_d$, but this does not always hold
in general.  However, the LaSalle invariance principle~\cite{LaSalle}
ensures that every solution $\rho(t)$ under (\ref{eqn:Lya-control})
converges to a set, often known as the LaSalle invariant set, and it
has been shown that the control design above renders the target state
$\rho_d$ almost globally attractive if (i) $H_0$ is strongly regular and
(ii) $H_1$ is fully connected~\cite{two-atoms}. Condition (i) requires
$H_0$ has distinct transition frequencies between any pair of energy
levels in the smallest invariant subspace $\H_s$ containing $\rho_0$ and
$\rho_d$.

\begin{figure}
\includegraphics[width=\columnwidth]{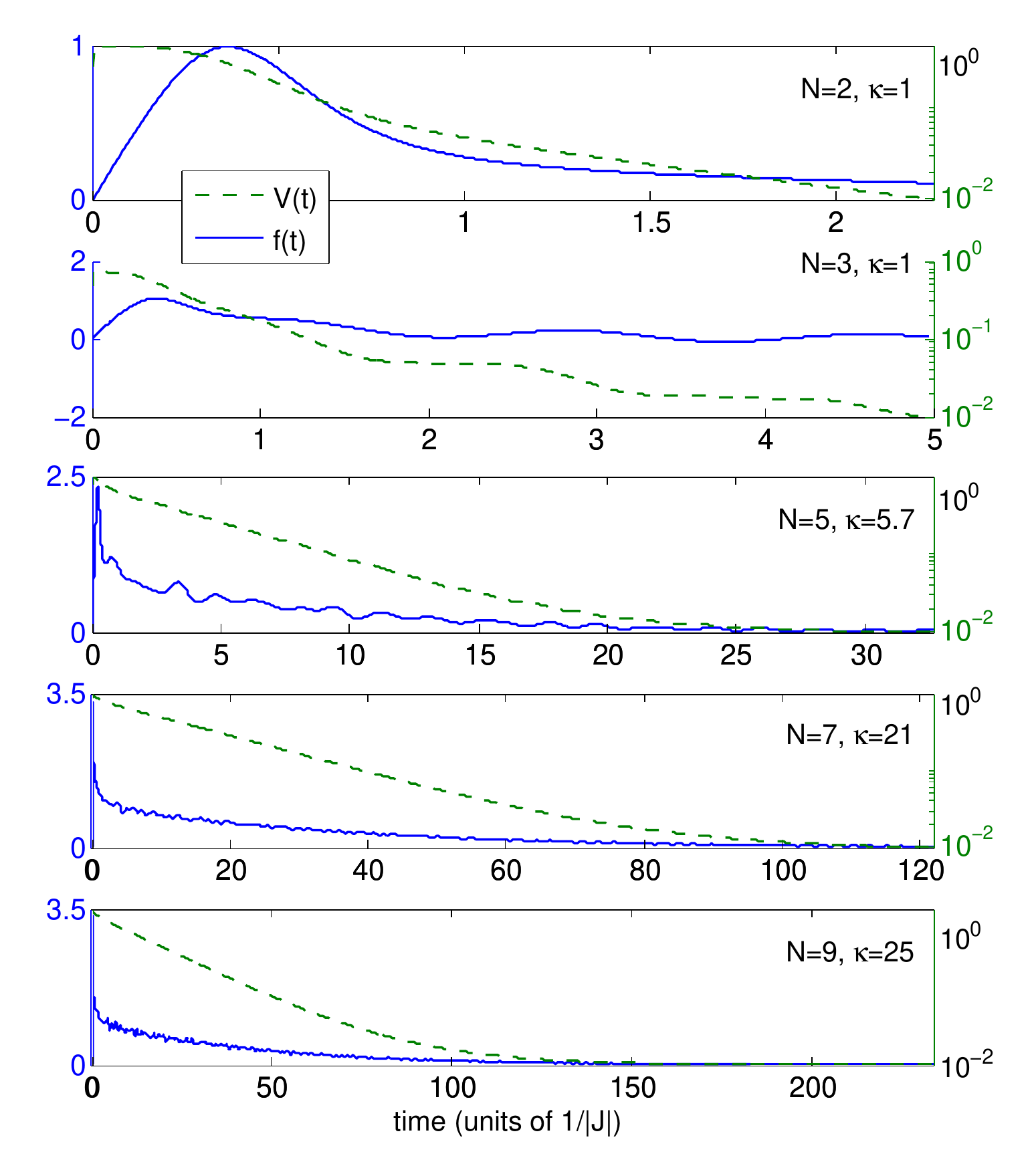} \caption{(Color
online) Examples of Lyapunov controls (solid lines) and ``distance''
$V(\rho(t),\rho_d)$ (dashed) from the GHZ state $\ket{\psi_d^{(1)}}$ for
Ising chains (\ref{eqn:Hf}) with $J<0$, of different length $N$,
assuming the system is initially prepared in the ground state
$\ket{\psi_0^{(f_0)}}$ with $f_0=10$.  The Lyapunov control pulses
achieve the desired state transfer in much shorter time than the
adiabatic controls for small $N$, and the approximately linear decrease
of the distance on the logarithmic scale, suggests that the distance
decreases exponentially.  For longer chains, however, the feedback
strength $\kappa$ has to be increased and the effects of ``attractive''
limit cycles interfere with convergence, resulting in control pulses
comparable or even longer than the corresponding adiabatic transfer
pulses.}
\label{fig:Lyapunov}
\end{figure}

For $N=2$ spins, with the initial and the target states chosen as
before, it can be verified that the two conditions above hold for the
dynamics restricted to $\H_s$, and the GHZ state $\rho_d$ is almost
globally attractive on $\H_s$.  Hence the Lyapunov control pulse
(\ref{eqn:Lya-control}) will steer the system from $\ket{\psi_0^+}$ to
$\rho_d$.  Moreover, as demonstrated in~\cite{two-atoms}, the Lyapunov
control design ensures that $\rho(t)$ converges ``exponentially'' to
$\rho_d$, and is insensitive to timing errors---unlike geometric control
schemes that require precise on-off switching of the control fields and
amplitude control, for instance.  Unfortunately, for longer Ising
chains, convergence of $\rho(t)$ to $\rho_d$ is no longer assured, as
the energy levels of $H_0$ on $\H_s$ are equally spaced, i.e., condition
(i) does not hold, and for $N\ge 4$ full connectivity is lost.
Theoretical analysis~\cite{IEEE_paper} shows that in this case the
invariant set is large, and most solutions converge to ``limit cycles''
a finite distance from the target state $\rho_d$. Nonetheless,
simulations suggest that we can ensure $\rho(t)$ with
$\rho(0)=\ket{\psi_0^{(f_0)}}\bra{\psi_0^{(f_0)}}$ converges to a point
very close to $\rho_d$ by carefully tuning the so-called feedback
strength $\kappa$ in (\ref{eqn:Lya-control}), as illustrated in
Fig.~\ref{fig:Lyapunov}.  We find that (i) for a given $N$ the final
fidelity achieved usually increases with $\kappa$; (ii) for a given
$\kappa$ the final fidelity decreases with increasing $N$.  This can be
explained by the fact that the dimension of the center manifold
surrounding $\rho_d$ increases rapidly with $N$.  Moreover, to achieve
higher fidelities for larger $\kappa$, the control time increases
rapidly as well.  Therefore, considering the transfer time required to
achieve a certain target fidelity close to $1$, Lyapunov control has a
significant edge for small $N$, but the advantage disappears for longer
chains (see Fig.~\ref{fig:Lyapunov}).  Thus, the Lyapunov control design
is only effective for short Ising chains.

\subsection{Optimal Control}

Lyapunov control can be considered as a kind of optimal control in that
the distance from the target state is monotonically decreasing with
time.  This form of optimal control is sometimes referred to as local
optimal control because at each point in time the control design is
based only on information about the current state of the system and the
target state.  From a computational point its main advantage is that the
control pulses can be calculated directly in a non-iterative fashion,
but as the last section shows, for more complex problems such as longer
chains this approach is not sufficiently powerful.  An alternative
method is to take a global approach, specify a target time $t_f$, and
attempt to maximize the fidelity $F(t_f)$ by globally varying the
control pulse $f(t)$ in the allowed control function space.  In
practice, optimal control problems over function spaces can generally be
solved only numerically, by parameterizing or discretizing the control
$f(t)$. The simplest and most common approach is to subdivide the time
interval $[0,t_f]$ and approximate the control $f(t)$ by a constant on
each subinterval $I_k$. This results in an optimization problem over
$\RR^K$, where $K$ is the number of time intervals, which can be solved,
e.g., by starting with an initial trial field $f^{(0)}(t)$ and
iteratively refining $f^{(n)}(t)$ such that the fidelity at the target
time $F^{(n)}(t_f)$ monotonically increases as a function of the
iteration $n$.  The crucial part of this procedure is the way
$f^{(n)}(t)$ is updated in each iteration.  We adopted here a
quasi-Newton method developed by Boyden, Fletcher, Goldfarb, and
Shanno~\cite{BFGS}.  A more detailed discussion of the optimization
process can be found in~\cite{chain-optimal}.

The adiabatic and the Lyapunov control pulses found earlier provide
upper bounds on $t_f$ to reach a certain final fidelity $F<1$, and our
previous reachability considerations for Ising chains with Hamiltonian
$H_f(t)$ suggest that there exists a solution $f(t)$ with $F(t_f))=1$
for a finite time $t_f$, although the proof does not give any hint as to
the best driving field $f(t)$ or how to discretize the control pulse.
Although it can be shown, roughly, that if the target state is reachable
by some admissible control field then there also exists a piecewise
constant control that achieves the same task, the theorem again does not
tell us how many time steps are needed.  If the resolution is too low,
i.e., $\Delta t=t_f/K$ is too large, the target fidelity may not be
achievable.  An interesting question for optimal control therefore is
how fast we can hope to achieve transfer to the target state, and what
time resolution of the field is required.  Can we do better than
adiabatic control?

\begin{figure}
\includegraphics[width=\columnwidth]{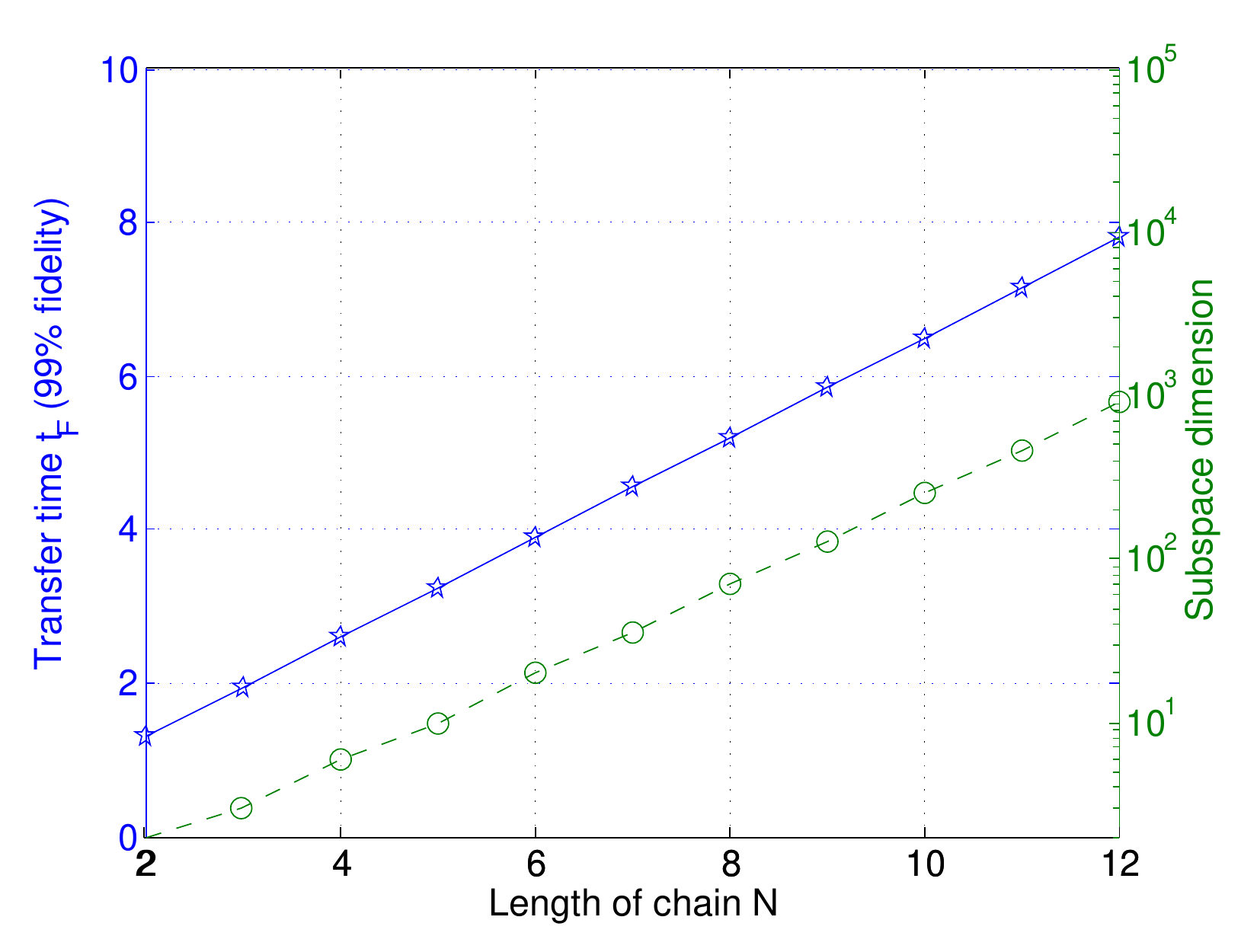} \caption{(Color
online) Lower bounds on the time $t_f$ (solid line; y-axis linear
scale) required to achieve 99\% transfer fidelity as a function of
chain length for $N=2$ to $N=12$, derived from optimal control
simulations.  Note the linear dependence of $t_f$ on $N$, despite an
exponential increase in the dimension of the relevant subspace
$\H_s$ (dashed line; y-axis logarithmic scale).}
\label{fig:opti-mini-time}
\end{figure}

Before we address these questions a final issue that needs to be
considered is constraints.  For example, the amplitude of the control
pulses will usually be limited by what can be achieved experimentally.
Whether to impose constraints and what kind depend on the details about
the system and the implementation.  For typical values in NMR
experiments $J$ is usually a few hundred Hz, while we can easily achieve
fields up to 50kHz for liquid-state NMR and is a few hundred kHz for
solid-state NMR~\cite{Chuang}, which would suggest a reasonable upper
bound on the field amplitudes of perhaps $10^3|J|$.  Simulations for
this problem suggest that bounds of this magnitude can usually be
neglected as the unconstrained optimization solutions almost always
satisfy these constraints, and unconstrained optimization is
computationally more efficient than constrained optimization.

Using an unconstrained optimization based on the above-mentioned
quasi-Newton method, we calculated the optimal pulses for Ising chains
of length $N=2$ to $12$.  Assuming $J<0$ and starting with the initial
state $\ket{\psi_0^+}$, we calculated optimal controls for different
values of the target time $t_f$ and time steps $K$, choosing the GHZ
state $\ket{\psi_d^{(1)}}$ as a target state.  For chains up to $N=12$
we were able to find controls that achieve at least $99$\% fidelity for
$t_f(N)=0.65N$ and $K=3N$, and in general only one or two runs with
either $f^{(0)}_k=1$ or a random number in $[0,1]$ for $k=1,\ldots,K$
were needed for the optimization to succeed in finding a suitable
control.  This shows that the transfer times for the optimal controls
are much shorter than those for the adiabatic pulses -- by about one
order of magnitude -- and the transfer time appears to be increase
linearly with the chain length $N$, at least up to $N=12$ as shown in
Fig. \ref{fig:opti-mini-time}.  This is quite surprising when one
considers that the dimension of the smallest invariant subspace $\H_s$
that contains the GHZ state $\ket{\psi_d^{(1)}}$ increases exponentially
in $N$.  A possible explanation for this linear dependence lies in the
dimension of the reachable set.  The dynamical Lie algebra on the
subspace $\H_s$ appears to be a high-dimensional representation of
$\su(N)$ or $\uu(N)$, suggesting that the reachable set starting with a
pure initial state is the homogeneous space
$\UU(N)/[\UU(1)\otimes\UU(N-1)]$, which has dimension
$N^2-1-(N-1)^2=2N-2$~\cite{Schirmer-orbits}.  This suggests that the
reachable set is a ($2N-2$)-dimensional manifold, which may explain the
apparently linear dependence in the observed bounds on the transfer
times despite the exponential increase in the dimension of the subspace
$\H_s$ it is embedded in (Fig. \ref{fig:opti-mini-time}).

\begin{figure}
\includegraphics[width=\columnwidth]{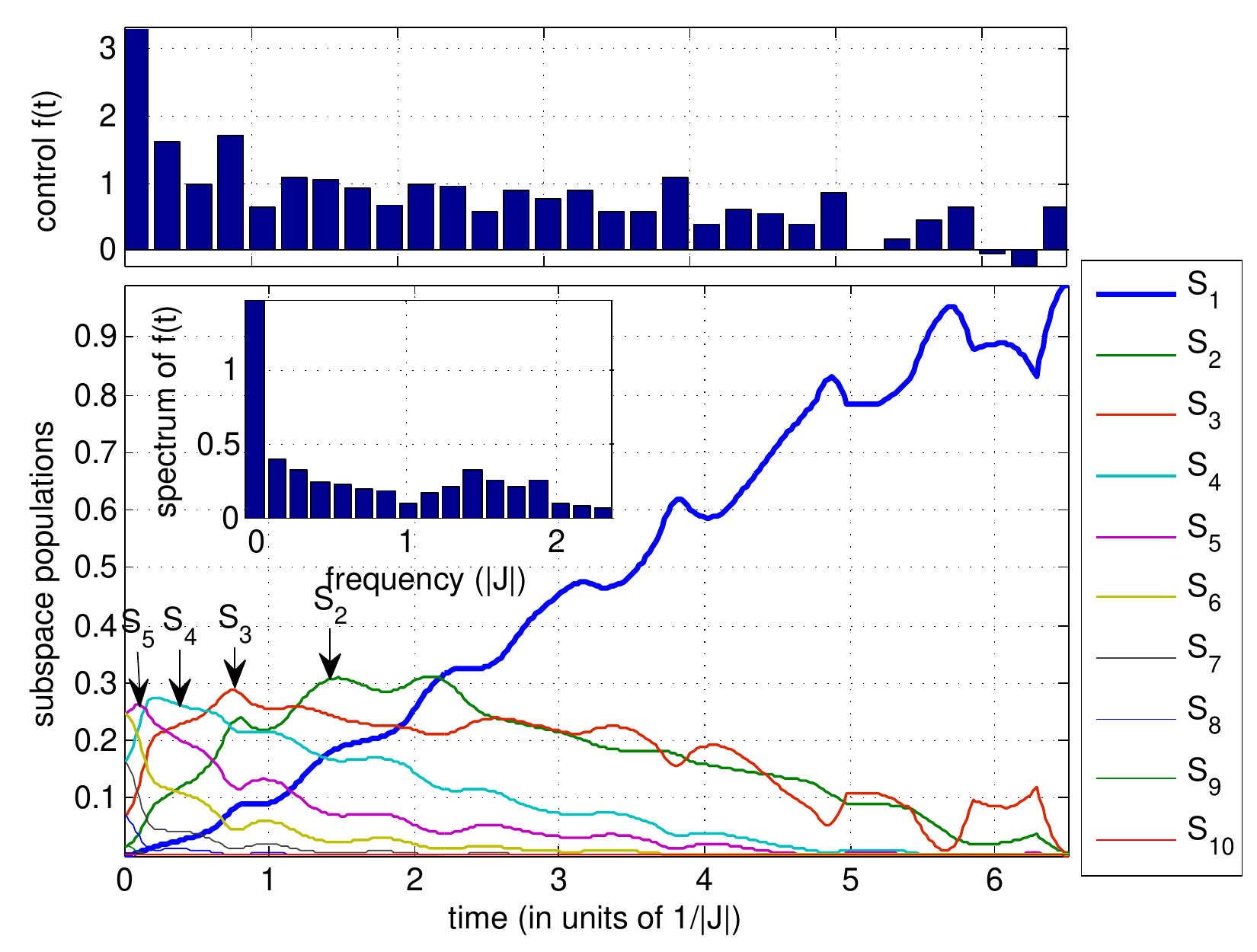}
\caption{ (Color online) Optimal control pulse (top) and evolution of the
ground state population $S_1$ (or fidelity $F(t)$) and other excited states 
population ($S_2, S_3,..., S_{10}$) for an Ising chain of length $N=10$
with $J<0$, starting in the state $\ket{\psi_0^+}$. The fidelity
corresponds to the projection onto the target GHZ state
$\ket{\psi_d^{(1)}}$. Inset shows spectrum of optimal pulse.}
\label{fig:GHZ-opti}
\end{figure}

Fig.~\ref{fig:GHZ-opti} shows an example of a typical optimal control
field $f(t)$ for a chain of length $N=10$ and the corresponding
evolution of the system.  The pulse appears well-behaved and quite
feasible both in time and frequency domain.  The corresponding evolution
of the system shows that the fidelity does not increase monotonically,
as is the case for an ideal adiabatic or Lyapunov control pulse.  The
figure also shows the evolution of the population of the eigenspace of
$H_0$.  For $N=10$, $H_0$, restricted to the relevant subspace $\H_s$ of
dimension 252, has $10$ distinct eigenvalues with corresponding
eigenspaces $S_{n}$ of varying dimensions.  The population of the 1D
ground state manifold $S_1$ corresponds to the fidelity $F(t)$.
Analysis of the population evolution shows that the increase in the
fidelity $F(t)$, or the ground state population $S_1$, is preceded by an
increase in the populations of the first and second excited state
manifold, and these increases are in turn preceded by peaks of the
populations of the eigenspaces $S_4$, $S_5$ and $S_6$, in this order.
This behavior can be explained in terms of the subspace coupling induced
by the interaction Hamiltonian $H_1$.  For the dynamics restricted to
the subspace $\H_s$, $H_1$ couples each eigenspace of $H_0$ only with
its first and second neighbor, i.e., the ground state manifold is
directly coupled only to the first and second excited state manifold,
and so forth.  Therefore population cannot be directly transferred from
e.g., $S_{10}$ to $S_{1}$, but must pass through several intermediate
levels.  Note that the maximum number of intermediate levels to be
traversed is linear in $N$ as the restriction of $H_0$ to the subspace
$\H_s$ containing the GHZ state $\ket{\psi_d^{(1)}}$ for a chain of
length $N$ has $N$ distinct eigenvalues.  If the system starts in the
state $\ket{\psi_0^+}$ or a state close to it, all of these eigenspaces
are initially populated.  This does not fully explain the behavior of the
optimized dynamics as there are many different excitation pathways and
optimal control attempts to maximize constructive interference between
all different paths leading to the desired outcome, but it suggests an
alternative explanation for the apparently linear dependence of the 
minimum transfer time on $N$, at least for chains up to length $N=12$, 
despite the exponential increase in the dimension of $\H_s$ from $2$ for 
$N=2$ to $924$ for $N=12$.

\section{Effect of Imperfections}

So far it was assumed that we have an ideal chain with uniform couplings
and can perfectly initialize the system in a pure separable state such
as $\ket{\psi_0^+}$ or the ground state of $\ket{\psi_0^{(f_0)}}$ for a
fixed large value of $f_0$ by cooling, and system inhomogeneity and the
effect of interactions with an environment were neglected.  To assess
the robustness of various control strategies with regard to such
imperfections in the context of this problem we shall consider how the
effectiveness of the control fields they produce for ideal systems is
diminished by various imperfections.

\subsection{Imperfect Initialization}

\begin{figure}
\includegraphics[width=\columnwidth]{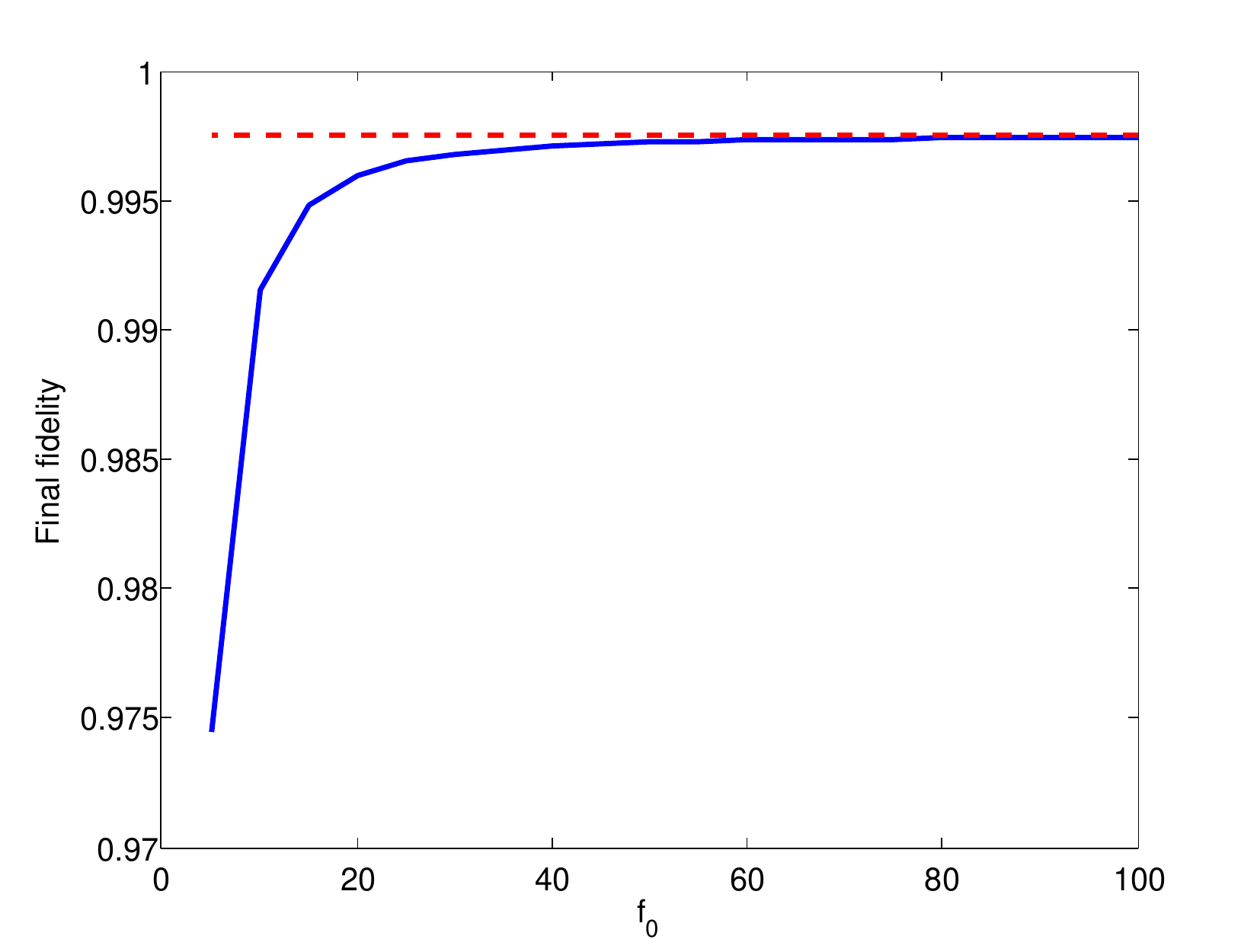} \caption{(Color
online) Final fidelity achieved for control pulse $f(t)$, which is
optimal for $J<0$ and the initial state $\ket{\psi_0^+}$, while the
actual initial state is the ground state $\ket{\psi_0^{(f_0)}}$, as a
function of $f_0$.  The error between the ideal (dashed line) and real
fidelity (solid line) goes to zero as $f_0$ increases.}
\label{fig:ini_error}
\end{figure}

Preparing the system in the ideal state $\ket{\psi_0^+}$ is not always
an easy task since we need to achieve $Jf_0\to-\infty$.  We have
already discussed using a finite $f_0$ for adiabatic passage and it was
shown that provided $f_0$ is sufficiently large, we always can realize
one of the GHZ states~(\ref{eqn:ghz-state}).  For Lyapunov and optimal
control strategies we may have been told that the system is initialized
in the state $\ket{\psi_0^+}$, while it was really initialized in the
ground state $\ket{\psi_0^{(f_0)}}$ for some finite value of $f_0$.  In
this case $f(t)$ is optimal for the initial state $\ket{\psi_0^+}$ and
using this field for the real initial state $\ket{\psi_0^{(f_0)}}$ does
not give us the maximal fidelity.  The effect of such errors is easy to
analyze.  Assume we have an ideal (optimal, adiabatic) control $f(t)$
that gives rise to an evolution $U(t)$ such that
$U(t_f)\ket{\psi_0^+}=\ket{\psi_d^{(1)}}$ for some finite $t_f$ or
$t_f\to +\infty$.  Let $\Pi_0=\ket{\psi_0^+}\bra{\psi_0^+}$ be the
projector onto $\ket{\psi_0^+}$. Then we have
\begin{align*}
  \ket{\psi_0^{(f_0)}}
  = \Pi_0\ket{\psi_0^{(f_0)}}+\Pi_0^\perp\ket{\psi_0^{(f_0)}}
  = c_0 \ket{\psi_0^+}+\Pi_0^\perp\ket{\psi_0^{(f_0)}}
\end{align*}
with $c_0=\ip{\psi_0^+}{\psi_0^{(f_0)}}$.  Furthermore
\begin{align*}
  \ip{\psi_d^{(1)}}{U(t_f)|\psi_0^{(f_0)}}
=& c_0\ip{\psi_d^{(1)}}{U(t_f)|\psi_0^+} \\
 & +\ip{\psi_d^{(1)}}{U(t_f)\Pi_0^\perp|\psi_0^{(f_0)}} = c_0
\end{align*}
as we have $U(t_f)\ket{\psi_0^+}=\ket{\psi_d^{(1)}}$, and $U(t_f)$ is
unitary. So, the maximum transfer fidelity is simply given by $|c_0|^2$,
the overlap of the actual initial state with the assumed initial state,
as illustrated in Fig.~\ref{fig:ini_error}.  For adiabatic control any
deviation of the initial state from the ground of the Hamiltonian at
time $t=0$ will limit the maximum transfer fidelity with the bound given
by $1-|c_0|^2$.  For optimal control the error can be decreased by
determining the initial state more accurately, e.g., using system
identification techniques such as state tomography, and incorporating
this information in the optimization.

\subsection{Thermal Effects}

\begin{figure}
\includegraphics[width=\columnwidth]{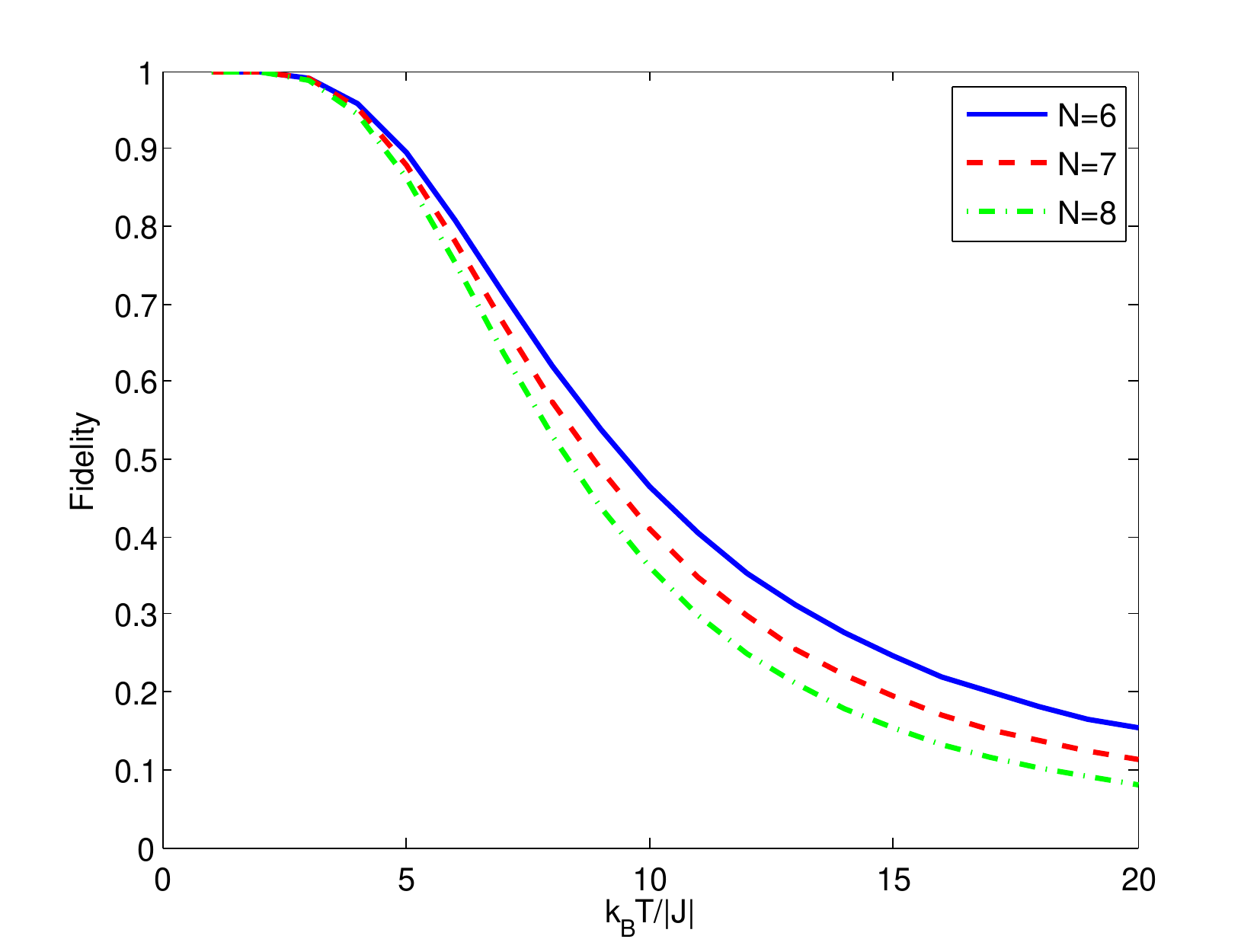}
\caption{(Color online) The final fidelity achieved by optimal
control for a thermal initial state (\ref{rho-thermal}) instead of a
pure state decreases with temperature, with a faster drop-off for
longer chains, as theoretically predicted, but high transfer
fidelities can still be achieved for reasonably low temperatures. }
\label{thermal-plot}
\end{figure}

Assuming the system is initialized by cooling in the presence of a
strong field in the $x$-direction, another source of error are thermal
effects, in particular the fact that zero temperature is not practically
achievable.  In general finite temperature effects result in an initial
state that is a thermal mixed state
\begin{equation}
  \label{rho-thermal}
  \rho(0) = \frac{e^{-H_f(0)/k_BT}}{\Tr[e^{-H_f(0)/k_BT}]},
\end{equation}
where $T$ is temperature and $k_B$ is the Boltzmann constant.  As the
populations of the eigenstates cannot change under adiabatic passage, we
see immediately that the maximum projection onto the target GHZ state we
can achieve is given by the ground state population
\begin{equation}
 \label{eqn:wn}
  w_1=\bra{\psi_0^{(f_0)}}\rho(0)\ket{\psi_0^{(f_0)}}.
\end{equation}
The initial population of the ground state depends both on the
temperature $T$ and the energy gap between the ground state and the
excited states of $H_f(0)$.  As the energy gap $\Delta \eps$ increases
roughly linearly with the field strength $f_0$, as we have seen earlier,
this explains why it is desirable to cool in the presence of a strong
field.

Based on the results in~\cite{chain-optimal} we might assume that the
transfer fidelity for optimal control could be improved by starting with
thermal initial state rather than the ground state.  Unfortunately, this
is not the case here because unlike in~\cite{chain-optimal} the
observable we are optimizing is a rank-1 projector onto a pure state,
$\ket{\psi_d^{(1)}}\bra{\psi_d^{(1)}}$, i.e., it has a single eigenvalue
of $1$ and all other eigenvalues are zero.  Therefore, the transfer
fidelity is bounded above by~\cite{Bounds}
\begin{equation}
  F(t_f)=\bra{\psi_d^{(1)}}\rho(t_f)\ket{\psi_d^{(1)}} \le w_1.
\end{equation}

Fig.~\ref{thermal-plot} shows the actual fidelities achieved by optimal
control for Ising chains of length $N=6,7,8$ with $J<0$ as a function of
temperature $T$, assuming the system is initialized in a thermal
ensemble of $H_f(0)$ with $f_0=10$.  The results are in line with the
expected decrease of the ground state population $w_1$ as a function of
the initial temperature $T$, and the fact that $w_1$ decreases faster
for longer chains.  In Fig.~\ref{fig:thermal_compare}, for Ising chain
with $N=6$ and $J<0$, assuming initialized in the thermal mixed state,
we have plotted how the final fidelity changes with respect to
temperature for the three different methods.  The relative flatness of
the curve for low temperatures suggests that all of these control
methods can achieve high-fidelity state transfer for reasonably long
chains and sufficiently low temperatures however, but the optimal
control pulses appear more robust with regard to thermal fluctuations
than the two other methods.

\begin{figure}
\includegraphics[width=\columnwidth]{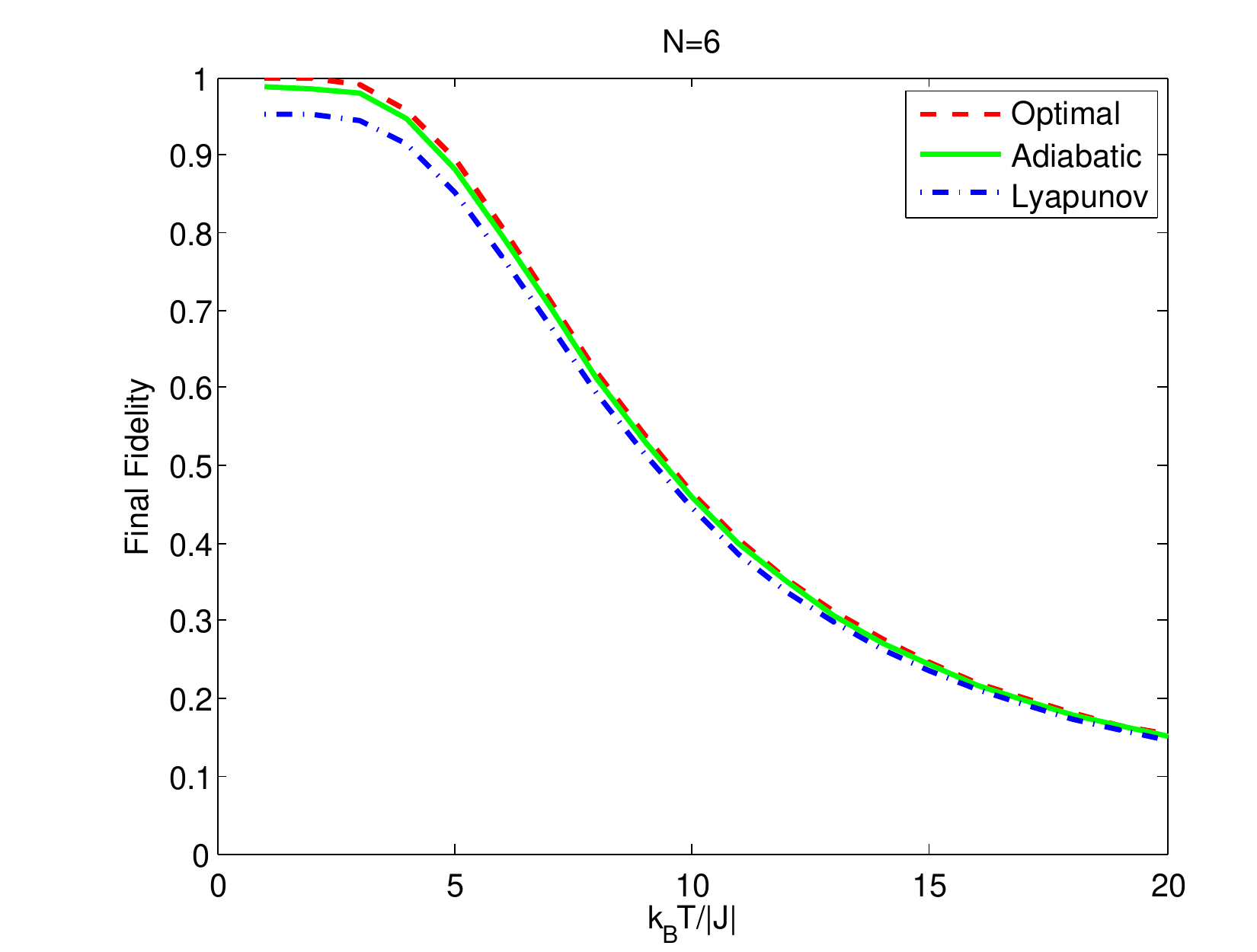}
\caption{(Color online) Final fidelities versus temperature for
different control methods for a chain of length $N=6$: (i) optimal
control, (ii) adiabatic control, and (iii) Lyapunov control, when the
control pulses generated assuming the system is initially in the ground
state of $H_0+f_0H_1$, are applied to initial states that are thermal
ensembles of different temperature.  The decrease in the final fidelity
as a function of the temperature is the similar for all three methods.}
\label{fig:thermal_compare}
\end{figure}

\subsection{Disordered Chains}

In reality it is impossible to make an absolutely uniform chain.  There
are always fluctuations in the spin coupling which make the chain
disordered.  To assess the effect of such inhomogeneity we consider the
case where the couplings between neighboring sites are randomly
perturbed around the average value $J$
\begin{equation} 
  \label{Ham_rand}
  H_f'(t)=\sum_{n=1}^{N-1}J(1+\delta_n) Z_nZ_{n+1}+f(t)\sum_{n=1}^{N}X_n,
\end{equation}
where $\delta_n\in [-\sigma,+\sigma]$ is a random variable with uniform
distribution around zero.  Since the random couplings are not known we
can only study their average effect over entanglement generation.  Thus,
we consider the evolution of the initial state $\ket{\psi_0^+}$ under
action on the perturbed $H_f'(t)$ for a pulse $f(t)$ optimized assuming
the Hamiltonian $H_f(t)$.  We repeat the experiment over 100 random
perturbations and then take the average value of the final fidelity over
all results.  Fig.~\ref{fig:rand_compare} shows the average fidelity as
a function of the parameter $\sigma$ of the fluctuations in the coupling
strength for different methodologies.  The figure shows that adiabatic
passage is very robust against disorder in the chain.  This is to be
expected as the ground state of the Hamiltonian $H_f'(t)$ is very
similar to $H_f(t)$ for all times when the fluctuations are small, and
therefore adiabatic passage is always able to steer the Ising chain from
$\ket{\psi^+}$ to the final GHZ state $\ket{\psi_d^{(1)}}$ independent
of the exact choice of the pulse $f(t)$.  Lyapunov and optimal control
pulses on the other hand rely on dynamic and interference effects to
achieve faster transfer, and the temporal shape of the optimal pulse is
thus more strongly dependent on the form of the Hamiltonian including
the coupling strength.  Thus, the optimal control pulses are more
susceptible to disorder although the optimal control pulses appear to
outperform adiabatic control for fluctuations up to $\sigma\approx
0.03$.  It is also interesting to note that the more efficient control
pulses obtained from global optimization techniques are more susceptible
than their less effective Lyapunov control cousins.  Again, the
performance of optimal control schemes can be significantly improved if
the actual couplings can be estimated more accurately using system
identification techniques~\cite{SysIdent} or closed-loop adaptive
strategies~\cite{adaptive}.  Alternatively, when accurate estimation of
the couplings is infeasible or impossible, e.g., because we wish to
control an ensemble of systems with slightly different couplings, the
performance of optimal pulses can (sometimes) be improved by optimizing
ensemble averages for a collection of systems as in \cite{Glaser},
although this may not be feasible for systems with more than a few spins
as optimizing over a collection of systems is computationally expensive,
with each function evaluation requiring the numerical solution of the
time-dependent Schrodinger equation for many systems.

\begin{figure}
\includegraphics[width=\columnwidth]{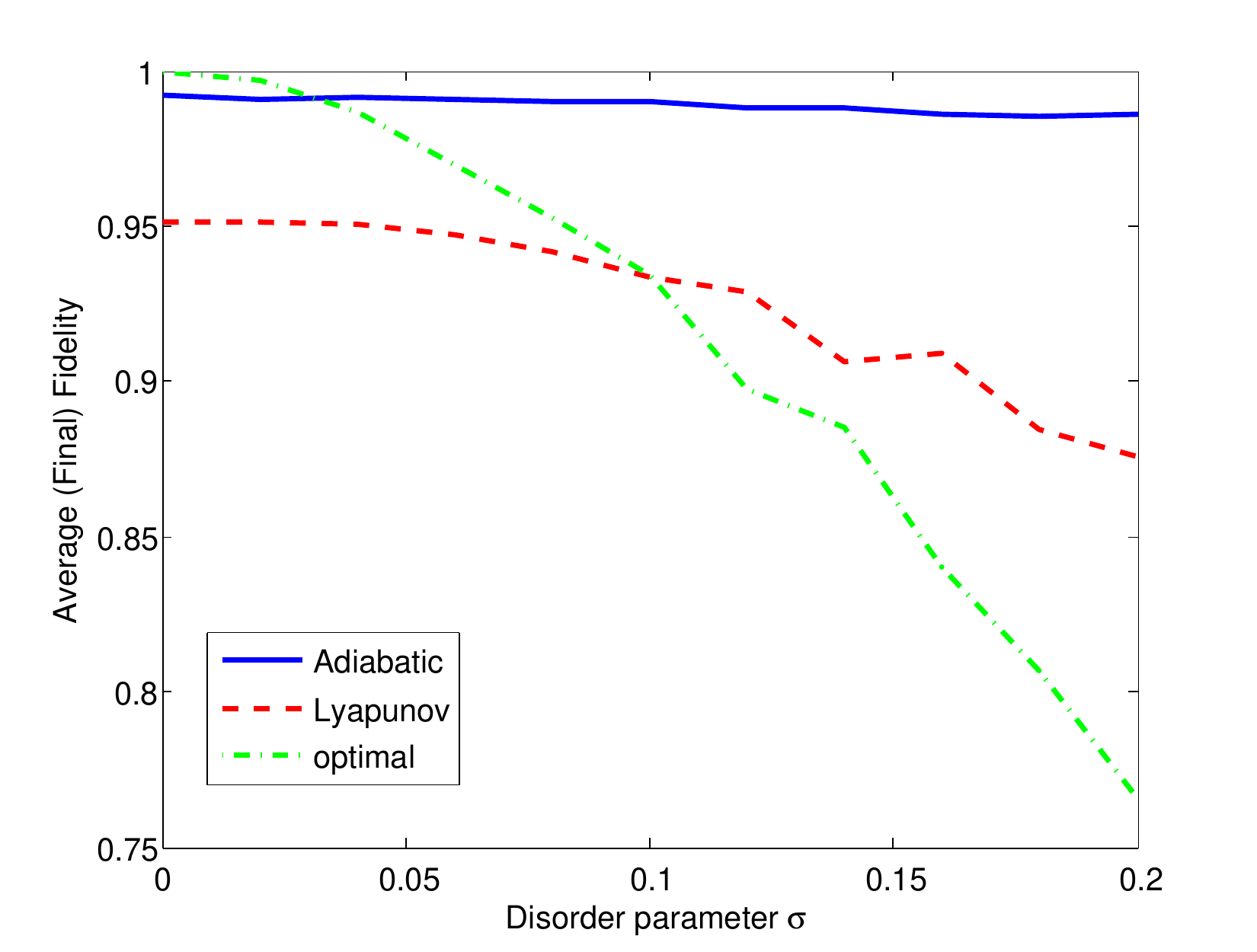}
\caption{(Color online) Comparison between different control methods for
GHZ generation for a chain of $N=6$ in terms of $\sigma$, the strength
of disorder. The target time was $t_f=50J^{-1}$ for adiabatic passage
and $t_f=15J^{-1}$ for Lyapunov and optimal control.}  \label{fig:rand_compare}
\end{figure}

\subsection{Decoherence Effect}

Another major problem in practice is that it is impossible to isolate
the system from its environment.  Any interaction with environment tends
to disturb the evolution of the system.  The precise effect of the
environment clearly depends on the type of interaction.  We shall assume
that the system-environment interaction is weak and Markovian and
modelled by a Lindblad equation
\begin{equation} 
 \label{Lindblad}
 \dot{\rho(t)}=-i[H_f(t),\rho(t)]+ \LL(\rho(t)),
\end{equation}
where $\LL(\rho)$ corresponds to dissipative effects and $\rho(t)$ is
the density matrix of the system.  We shall focus here on dephasing,
which destroys the coherence of the system, as it is common for many
physical systems and the typical dephasing times for most physical
systems are much shorter than other relaxation rates.

To assess the effect of decoherence we solve the Lindblad equation using
the appropriate adiabatic or optimized pulse $f(t)$ obtained for the
noiseless system.  For a dephasing noise we consider the specific model
\begin{equation}
  \label{Leinblad2}
    \LL(\rho(t))=-\gamma
    \sum_{n=1}^N\{\rho(t)-Z_n\rho(t) Z_n\},
\end{equation}
which corresponds to localized dephasing of individual spins.
Fig.~\ref{fig:dic_compare} shows the fidelity achieved as a function of
noise strength $\gamma$ for all methods.  As expected, the fidelity
decays as the noise strength increases for all strategies but adiabatic
passage is substantially more susceptible to the dephasing noise of the
form considered.  This is mainly due to the times involved: the transfer
times for the adiabatic pulses tend to be about one order of magnitude
greater than those for the optimal pulses, giving dephasing more time to
act and destroy the coherence.  Optimal control is more robust, and
theoretically, it may be possible to improve the performance of optimal
control by taking decoherence effects into account in the optimization
progress, although this is computationally demanding as it require full
density matrix optimzation for a density matrix of Hilbert space of
dimension $2^N$, as opposed to pure-state optimization on a subspace,
which can be done far more efficiently.

\begin{figure}
\includegraphics[width=\columnwidth]{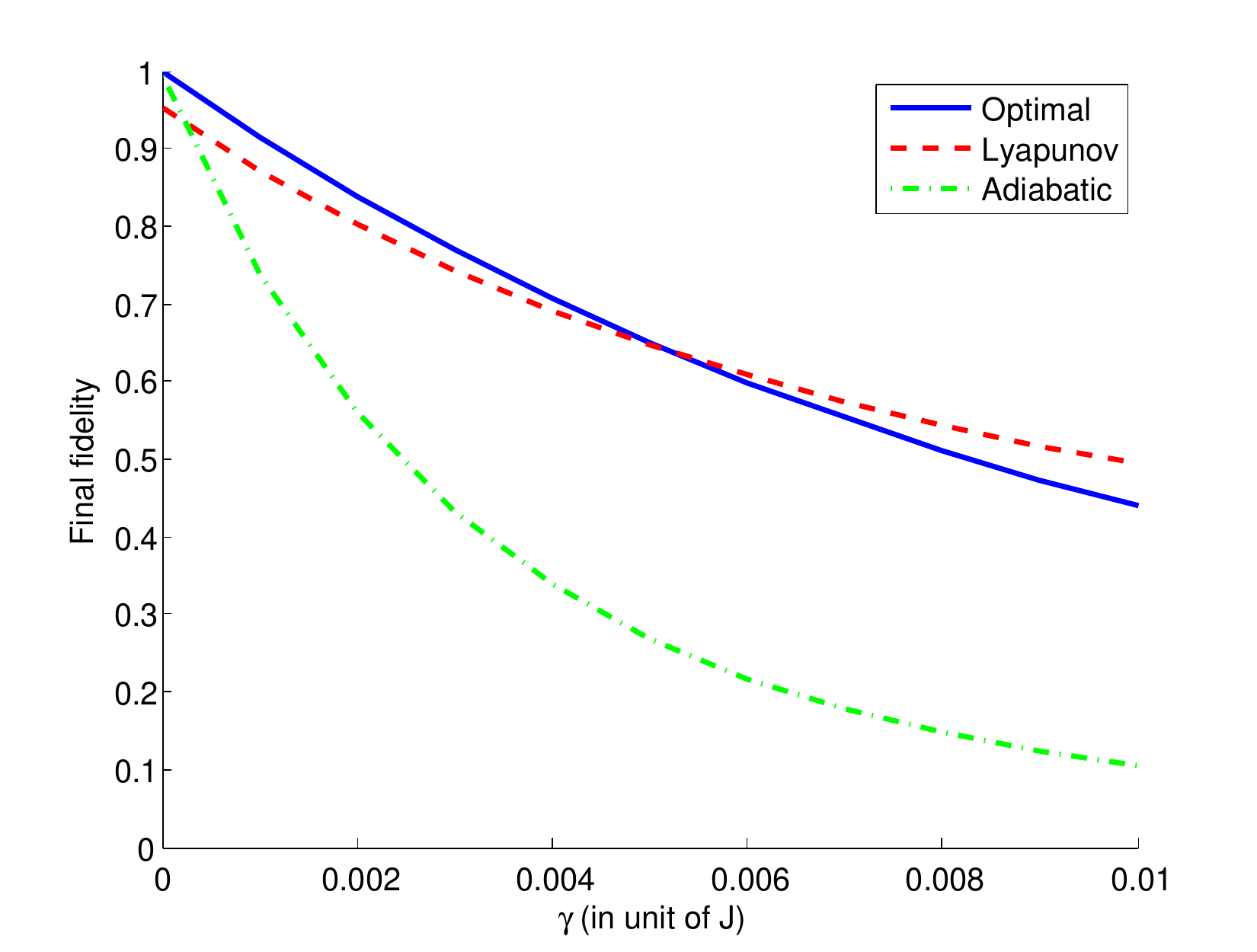} \caption{(Color
online) Effect of dephasing on different methods for GHZ generation for
$N=6$.  The longer pulses required for adiabatic passage leave it more
susceptible to dephasing.  $t_f=50J^{-1}$ for adiabatic passage,
$t_f=15J^{-1}$ for Lyapunov and optimal control.}
\label{fig:dic_compare}
\end{figure}

\section{Concluding discussion}

We have studied Ising chains subject to a single a global control such
as a magnetic field in the $x$-direction with regard to preparing them
in an entangled GHZ state.  Due to the existence of multiple symmetries
the system is not controllable, and the Hilbert space decomposes into
subspaces of varying dimensions that are invariant under the dynamics.
Our analysis further shows that even on the invariant subspaces the
dynamics is not controllable for all but the smallest subspaces for
$N>3$, and that the GHZ state of interest always lies in one of the
largest invariant subspaces on which the dynamics is not controllable
except for $N=2,3$.  Nonetheless the GHZ state is (asymptotically)
reachable from an easy-to-prepare initial state by adiabatic passage,
and under certain conditions finite-time reachability from a product
state such as $\ket{+\ldots+}$ can be inferred for chains of arbitrary
length $N$.  Motivated by this positive reachability result, three
different control strategies---adiabatic control, Lyapunov control and
optimal control---were considered to prepare a GHZ state, and their
advantages and disadvantages discussed, serving as a reference for
future experimental implementations.

Each method has its own advantages and disadvantages.  Broadly speaking
Lyapunov control is a simple form of optimal control design, which
produces pulses that are simple and quite effective for short chains but
the method struggles for longer chains, and for $N>4$ both adiabatic
and global optimal control seem to be generally superior.  In terms
of the total time required to prepare a GHZ state with high fidelity for
a given chain of length $N$, the control pulses found by global optimal
control proved the most efficient, with the time required being up to an
order of magnitude less than the time-scale for adiabatic transfer.  The
short pulse durations also confer greater robustness in the presence of
local spin dephasing compared to the adiabatic transfer scheme, and the
optimal control pulses appear slightly more robust with regard to finite
temperature effects.  The adiabatic transfer scheme, on the other hand,
is more robust with regard to system inhomogeneity such as unknown
fluctuations in the $J$-coupling between spins.  This is to be expected
considering that small system perturbations do not change the ground
state of the system appreciably, while such perturbations can alter the
interference between different excitation pathways that optimal control
designs aim to exploit.  The performance of the latter can be improved
by system identification or closed-loop adaptive strategies, and in 
some cases by explicitly taking inhomogeneity or decoherence into account
in the optimization.

Both the adiabatic pulse and the Lyapunov pulse have analytic
expressions, and especially for the adiabatic control, we can choose the
form of the slowly varying pulse without solving the dynamical equation.
It is interesting in this regard to note that although a linearly
decreasing field would appear to be the simplest choice for adiabatic
control, simulations and analysis suggest that a field of the form $f_0
e^{-\mu t}$ is generally preferable, resulting in both faster transfer
and increased robustness.  While the pulse shapes of the optimal pulses
are rather random and more complex than the corresponding adiabatic
control fields, the amplitude range and spectral bandwidth of the pulses
are quite narrow, and certainly appear to be within experimentally
accessible limits.  One major drawback of optimal control is that,
unlike for the adiabatic and Lyapunov control designs, there are no
explicit expressions for the optimal pulse.  Instead the pulse must be
computed by numerically, and the complexity and computational overhead
of calculating the optimal pulses increases rapidly with the length of
the chain.  Although the dynamics can be restricted to a subspace,
unfortunately, the GHZ state of interest belongs to one of the largest
invariant subspaces $\H_s$, whose dimension increases exponentially in
$N$ with $\dim \H_s \approx 2^N/4$ for large $N$, which substantially
increases the computational complexity.  Surprisingly, despite this
exponential increase of the subspace dimension, the minimum time
required to prepare a GHZ state given an initial product state such as
$\ket{+\ldots+}$ using optimal control appears to be actually linear in
$N$. Such fast production of a multi-spin GHZ (or Cat or NOON) state
using a minimal global field on an Ising chain will be highly valuable
in the fields of enhanced sensing and quantum information technology.

\emph{Acknowledgments:} SGS acknowledges funding from EPSRC ARF Grant
EP/D07192X/1, the QIPIRC and Hitachi.  XW thanks the Cambridge Overseas
Trust, Hughes Hall and the Cambridge Philosophical Society for support.
SB is supported by EPSRC ARF grant EP/D073421/1, through which AB is
also supported.  SB also acknowledges the Royal Society and the Wolfson
Foundation.

\end{document}